\documentclass[num-refs]{wiley-article}
\usepackage{wrapfig}

\usepackage{siunitx}
\usepackage{framed}

\usepackage{stackengine}

\newcommand\baron[1]{\stackon[1pt]{$#1$}{\rule{.9ex}{.001ex}}}

\usepackage{autobreak}
\allowdisplaybreaks

\def \B{\mathcal{B}}
\def \C{\mathcal{C}}

\def \E{\mathcal{E}}

\def \M{\mathcal{M}}

\def \P{\mathcal{P}}
\def \Q{\mathcal{Q}}
\def \R{\mathcal{R}}
\def \S{\mathcal{S}}
\def \T{\mathcal{T}}

\def \X{\mathcal{X}}
\def \Y{\mathcal{Y}}
\def \Z{\mathcal{Z}}

\def \fc{\mathbf{c}}
\def \fc{\mathbf{c}}

\def \fu{\mathbf{u}}

\def \fx{\mathbf{x}}
\def \fy{\mathbf{y}}

\def \fp{\mathbf{p}}

\def \fY{\mathbf{Y}}

\def \f0{\mathbf{0}}

\def \frho{\boldsymbol{\rho}}
\definecolor{blau_1a}{RGB}{93,133,195}
\definecolor{blau_2a}{RGB}{0,156,218}
\definecolor{gruen_3a}{RGB}{80,182,149}
\definecolor{gruen_4a}{RGB}{175,204,80}
\definecolor{gruen_5a}{RGB}{221,223,72}
\definecolor{orange_6a}{RGB}{255,224,92}
\definecolor{orange_7a}{RGB}{248,186,60}
\definecolor{rot_8a}{RGB}{238,122,52}
\definecolor{rot_9a}{RGB}{233,80,62}
\definecolor{lila_10a}{RGB}{201,48,142}
\definecolor{lila_11a}{RGB}{128,69,151}

\definecolor{blau_1b}{RGB}{0,90,169}
\definecolor{blau_2b}{RGB}{0,131,204}
\definecolor{gruen_3b}{RGB}{0,157,129}
\definecolor{gruen_4b}{RGB}{153,192,0}
\definecolor{gruen_5b}{RGB}{201,212,0}
\definecolor{orange_6b}{RGB}{253,202,0}
\definecolor{orange_7b}{RGB}{245,163,0}
\definecolor{rot_8b}{RGB}{236,101,0}
\definecolor{rot_9b}{RGB}{230,0,26}
\definecolor{lila_10b}{RGB}{166,0,132}
\definecolor{lila_11b}{RGB}{114,16,133}

\definecolor{mycolor1}{rgb}{0.0, 0.18, 0.39}
\definecolor{mycolor2}{RGB}{87,108,67}
\definecolor{mycolor3}{RGB}{8,133,161}
\definecolor{mycolor4}{RGB}{80,91,161}
\definecolor{mycolor5}{RGB}{98,122,157}
\definecolor{mycolor6}{RGB}{255,163,67}
\definecolor{mycolor7}{RGB}{152,205,225}
\definecolor{mycolor8}{RGB}{242,204,48}
\definecolor{mycolor9}{rgb}{0,.5,0}
\definecolor{mycolor10}{rgb}{.59,.44,.09}
\definecolor{mycolor11}{RGB}{231,199,31} 
\definecolor{mycolor12}{RGB}{8,133,161} 
\definecolor{mycolor13}{RGB}{157,188,64} 
\definecolor{mycolor14}{RGB}{194,150,130} 
\definecolor{mycolor15}{RGB}{98,122,157} 
\definecolor{mycolor16}{RGB}{160,160,160} 
\definecolor{mycolor17}{RGB}{115,82,68} 
\definecolor{mycolor18}{RGB}{94,60,108} 
\definecolor{mycolor19}{RGB}{115,82,68} 
\definecolor{mycolor20}{RGB}{255,183,30} 

\usepackage{commath,amsmath,amssymb,amsfonts}
\usepackage{mathtools}
\usepackage{algorithmic}
\usepackage{pgfplots} 
\usetikzlibrary{arrows}
\pgfplotsset{compat=1.15}
\usepackage{graphicx}
\usepackage{xfrac}
\usepackage{colortbl}
\usepackage{cancel}
\usepgfplotslibrary{fillbetween}
\usepackage{bm}
\usepackage{float}
\usepackage{hyperref}
\usepackage{multirow}
\usepackage{xcolor}
\usepackage{mathrsfs}
\usepackage{bbm}
\usepackage{amssymb}
\usepackage{autobreak}
\allowdisplaybreaks
\usepackage{tikz}
\usetikzlibrary{arrows.meta}
\usepackage{amsbsy}
\usepackage{bm}
\usepackage{fixmath}
\usepackage{silence}
\usetikzlibrary{math}
\usepackage{commath,amsmath,amssymb,amsfonts}
\usepackage{mathtools}
\usepackage{pgfplots}
\pgfplotsset{compat=1.15}
\usepackage{graphicx}
\usepackage{pgfgantt}
\usepackage{pdflscape}

\usepackage{upgreek}
\usepackage{nicematrix}

\usepackage{leftindex}
\let\svtikzpicture\tikzpicture
\def\tikzpicture{\noindent\svtikzpicture}

\usepackage{wasysym}
\DeclareMathSymbol{\mhexagon}{\mathord}{wasy}{57}

\papertype{ArXiv Preprint}
\paperfield{Molecular Communication}

\title{Deterministic Identification For MC ISI-Poisson Channel}

\author[1]{Mohammad Javad Salariseddigh}
\author[2]{Vahid Jamali}
\author[3]{Uzi Pereg}
\author[4]{Holger Boche}
\author[5]{Christian Deppe}
\author[6]{Robert Schober}

\affil[1]{Institute For Communication Engineering, Technical University of Munich, Munich, 80333, Germany, E-mail: mjss@tum.de}
\affil[2]{Group of Resilient Communication Systems, Technical University of Darmstadt, E-mail: vahid.jamali@tu-darmstadt.de}
\affil[3]{Department of Electrical and Computer Eng., Technion – Israel Institute of Technology, Haifa, 80333, Israel, E-mail: uzi.pereg@tum.de}
\affil[4]{Chair of Theoretical Information Technology, Technical University of Munich, Munich, 80333, Germany, E-mail: boche@tum.de}
\affil[5]{Institute For Communication Engineering, Technical University of Munich, Munich, 80333, Germany, E-mail: christian.deppe@tum.de}
\affil[6]{Institute for Digital Communications, Friedrich-Alexander-University Erlangen-N{\"u}rnberg, Erlangen, Germany, E-mail: robert.schober@fau.de}

\corraddress{Mohammad Javad Salariseddigh, PhD Researcher}
\corremail{mjss@tum.de}
\fundinginfo{
\vspace{.8mm}
\begin{enumerate}
    \item 6G-Life Project, GNT: 16KISK002 
    \item DFG, GNT: JA 3104/1-1
    \item Israel CHE Prog. QD. Sci. Tech., GNT: 86636903
    \item BMBF, GNT: 16KIS1003K; MAMOKO, GNT: 16KIS0914
    \item BMBF, GNT: 16KIS1005, 16KISQ028; 6G-Life Project, GNT: 16KISK002
    \item MAMOKO, GNT: 16KIS0913
\end{enumerate}
}
\runningauthor{M. J. Salariseddigh et al.}

\begin{document}

\begin{frontmatter}
\maketitle

\begin{abstract}
Several applications of molecular communications (MC) feature an alarm-prompt behavior for which the prevalent Shannon capacity may not be the appropriate performance metric. The identification capacity as an alternative measure for such systems has been motivated and established in the literature. In this paper, we study deterministic identification (DI) for the discrete-time \emph{Poisson} channel (DTPC) with inter-symbol interference (ISI) where the transmitter is restricted to an average and a peak molecule release rate constraint. Such a channel serves as a model for diffusive MC systems featuring long channel impulse responses and employing molecule counting receivers. We derive lower and upper bounds on the DI capacity of the DTPC with ISI when the number of ISI channel taps $K$ may grow with the codeword length $n$ (e.g., due to increasing symbol rate). As a key finding, we establish that for deterministic encoding, the codebook size scales as $2^{(n\log n)R}$ assuming that the number of ISI channel taps scales as $K = 2^{\kappa \log n}$, where $R$ is the coding rate and $\kappa$ is the ISI rate. Moreover, we show that optimizing $\kappa$ leads to an effective identification rate [bits/s] that  scales linearly with $n$, which is in contrast to the typical transmission rate [bits/s] that is independent of $n$.
\end{abstract}
\end{frontmatter}

\section{Introduction}
Molecular communications (MC) is a bio-inspired promising paradigm for communication between nanomachines or different biological entities, such as cells and organs \cite{Nakano13} and realizes the exchange of information via the transmission, propagation, and reception of signaling molecules \cite{NMWVS12,FYECG16}. In the past decade, different aspects of synthetic MC have been explored in the literature from distinctive viewpoints, including channel modeling \cite{Jamali19,Jamali19_2}, modulation and detection design \cite{Kuscu19}, biological building blocks for transceiver design \cite{Soldner20}, and information-theoretical and relevant mathematical foundations \cite{Gohari16,Rose19,Hsieh13}. Moreover, several proof-of-concept implementations of synthetic MC systems have been reported in the literature, see, e.g., \cite{Farsad17,Giannoukos17,Unterweger18}. Furthermore, the ongoing progress in synthetic biology \cite{Grozinger19,Soldner20} is expected to enable sophisticated MC systems in the future, capable of performing the complex computation and communication tasks required for realizing the Internet of Bio-nano Things \cite{Aky15,Senturk22,Mcbride19,Liu17}. Also, the authentication problem \cite{Simmons88} which exhibit affinity to the identification problem is considered in \cite{Zafar19}.

In particular, one of the basic and widely-accepted abstract models for MC systems with molecule counting receivers is the discrete-time Poisson channel (DTPC) model with \emph{inter-symbol interference} (ISI) model \cite{Arjmandi13,Jamali18,Gohari16}. The DTPC model with memory has been used to study the performance limits of MC systems. Despite the recent theoretical and technological advancements in the field of MCs, the transmission capacity of most MC systems with DTPC with memory model are still unknown \cite{Gohari16}. However, a number of approaches to examining the behavior of Poisson channel are being explored. For instance, an analytic expression for the transmission capacity of a DTPC with memory under an average power constraint alone, is still open \cite{Gohari16,Cao13_1,Cao13_2}. However, several bounds and asymptotic behaviors for the DTPC with memory in different setups have been established. For instance, analytical lower and upper bounds on the transmission capacity of the DTPC with input constraints and memory are provided in \cite{Ratti21}. Bounds on the transmission capacity of the DTPC with memory are developed in \cite{Aminian15_2,Aminian15}. The design of optimal code for DTPC with memory under a peak and average power constraint is studied in \cite{Ahmadypour21}. In \cite{Mosayebi14}, the impact of memory on the performance for a diffusive MC channel is characterized. Performance analysis of modulation schemes for diffusive MC with memory is considered in \cite{Galmes16} and impact of degree of memory on the performance is shown. Design of the filter and detector parts in a receiver for Poisson channel with time-varying mean when transmitted symbols are exposed to the ISI is studied in \cite{Vakilipoor20}. The Code design problem for diffusive MC channel under ISI is considered in \cite{Kislal19} where influence of the ISI is incorporated into the code design. The authors in \cite{Salariseddigh22,Salariseddigh_GC_IEEE} studied the DTPC in absence of ISI, i.e., $K=1$, and established lower and upper bounds on the DI capacity where the codebook size scales as $\sim 2^{(n\log n)R}$.

Numerous applications of MC within the scope of future generation wireless networks (XG) \cite{6G+,6G_PST} are linked with event-triggered communication systems. In such systems, Shannon's message transmission capacity, as studied in \cite{Gohari16,Farsad18,Rose19,Pierobon12,Farsad20,Hsieh13,Aminian15_2,Aminian15,Etemadi19,Etemadi19_2,Mahdavifar15,Ratti21}, may not be the appropriate performance metric, instead, the identification capacity is deemed to be an essential quantitative measure. In particular, in object-finding or event-detection scenarios, where the receiver aims to determine the presence of an object or determine the occurrence of an specific event in terms of a reliable Yes\,/\,No answer, the so-called identification capacity is the key applicable performance measure \cite{AD89}. Concrete examples of the identification problem within the MC context include health monitoring \cite{Nakano14,ghavami2020anomaly} where, e.g., one may desire in whether or not the pH value of the cerebrospinal fluid of brain exceeds a crisis level; targeted drug delivery \cite{Muller04,Nakano13} and cancer treatment \cite{Jain99,Wilhelm16,Hobbs_ea98}, where, e.g., a nano-device's purpose is to identify whether or not an specific cancer biomarker exist in the vicinity of the target tissue. Moreover, identification problems can also be found in various natural MC systems. For instance, in bionic nose setting \cite{Liu22}, or in natural pheromone communications \cite{wyatt2003pheromones,kaupp2010olfactory} where, e.g., animals involved in mating seeks sexual pheromones to realize the presence of an opposite sex. In fact, the olfactory systems of animals have the capability of \textit{recognizing} the presence of extremely large numbers of different molecule mixtures (e.g., pheromones, odors, etc.) \cite{Buck05,Buettner17}, which has inspired researchers to regard them as role models for the design of bio-inspired synthetic MC systems \cite{Jamali22}. Motivated by this discussion, in this paper, we investigate the fundamental performance limits of identification problem in MC systems, which can be modelled by the DTPC with ISI.
\vspace{-5mm}

\subsection{Related Work on Identification Capacity}
In Shannon's communication paradigm \cite{S48}, message encoding is conducted by sender, Alice, in a certain way that guarantee a reliable recovering of the original message by the receiver, Bob. In contrast, for the identification setting, the coding scheme is designed to accomplish a different objective \cite{AD89}. Specifically, the decoder's task is to verify whether a \emph{particular} message was sent or not. Ahlswede and Dueck \cite{AD89} introduced a randomized-encoder identification (RI) scheme, in which the codewords are tailored according to their corresponding distributions. Although employing such distributions does not bring advantage in terms of gain in the Shannon's message transmission capacity \cite{A78} or codebook size, Ahlswede and Dueck \cite{AD89} established that providing local randomness at the encoder, reliable identification is yielded a remarkable attribute regarding the codebook size, namely, the codebook size exhibit a double-exponentially growth in the codeword length $n$, i.e., $\sim 2^{2^{nR}}$ \cite{AD89}, where $R$ is the coding rate. This observation is extremely different from the conventional message transmission problem, which has an exponential codebook size in the codeword length, i.e., $\sim{2^{nR}}$. The construction of randomized identification (RI) codes is considered in \cite{VK93,KT99,Gunlu21}. For example, a binary code using three-layer concatenated constant-weight code is established in \cite{Verdu93}. Nevertheless, the realization of such codes entails extra complexity \cite[see Sec.~1]{Salariseddigh22}.

In the deterministic encoding setting of identification, also referred to as deterministic identification (DI) \cite{Salariseddigh_IT} or identification without randomization \cite{AN99}, the codewords are selected by a deterministic function. In our recent works \cite{Salariseddigh_ITW,Salariseddigh_ICC,Salariseddigh_IT,Salariseddigh_GC_IEEE}, we target DI for channels with power constraint, including discrete memoryless channels (DMCs), Gaussian channels with fast and slow fading, and the memoryless discrete-time Poisson channel (DTPC), respectively. The codebook size of DI for DMCs, similar to the transmission problem grows exponentially in the codeword length \cite{AD89,AN99,Salariseddigh_ICC,J85,Bur00}, however, the achievable identification rates are significantly higher compared to the transmission rates \cite{Salariseddigh_ICC,Salariseddigh_IT}. For the Gaussian channel \cite{Salariseddigh_ITW,Salariseddigh_arXiv_ITW} and DTPC \cite{Salariseddigh_GC_IEEE,Salariseddigh_GC_arXiv}, the codebook size scales as $\sim 2^{(n\log n)R}$. Deterministic codes often have the advantage of simpler implementation, simulation \cite{Brakerski20,PP09}, and explicit construction \cite{A09}. DI problem for Gaussian channels is also studied in \cite{Labidi21,Wiese22,BV00,Salariseddigh_ITW}. Further, DI may be preferred over RI in complexity-constrained applications of MC systems, where the generation of random codewords is challenging\footnote{\,\textcolor{mycolor5}{On the other hand, we note that the biological hardware of MC systems (e.g., reaction networks) features an inherent stochastic nature \cite{Chou17}} \textcolor{mycolor5}{which can potentially be exploited for realizing RI.}}. Motivated by this discussion, in this paper, we investigate the fundamental performance limits of the DI problem in MC systems, which can be modelled by the DTPC with ISI.
\subsection{Contributions}
In this paper, we consider MC systems employing molecule counting receivers with a large number of released molecules at the transmitter, see \cite[Sec.~IV]{Jamali19}. Further, we assume that the received signal experiences ISI and follows the Poisson distribution. We formulate the problem of DI over the DTPC with memory under average and peak molecule release rate constraints to account for the limited molecule production\,/\,release rates of the transmitter. As our main objective, we investigate the fundamental performance limits of DI over the DTPC with ISI. In particular, this paper makes the following contributions:
\begin{itemize}
    \item[$\blacklozenge$] \textbf{\textcolor{blau_2b}{Generalized ISI Model}}: In MC systems, often the number of channel taps $K$ can be large, particularly for non-degrading signalling molecules in bounded environments, which leads to a long channel impulse response (CIR). In addition, the value of $K$ increases not only with the dispersiveness of the channel but also with the symbol rate. Therefore, it is of interest to investigate the asymptotic limits of the system for large symbol rates (leading to large $K$) and large codeword lengths $n$. To do so, we consider a generalized ISI model that captures the ISI-free channel (i.e., $K=1$), ISI channels with constant $K>1$, and ISI channels for which $K$ increases with the codeword length $n$ (e.g., due to increasing symbol rate). To the best of the authors' knowledge, such a generalized ISI model has not been studied in the literature, yet.
    \item[$\blacklozenge$] \textbf{\textcolor{blau_2b}{Codebook Scale}}: We establish that the codebook size of the DTPC with ISI for deterministic encoding scales in $n$ similar to the memoryless DTPC \cite{Salariseddigh22}, namely super-exponentially in the codeword length ($\sim 2^{(n\log n)R}$), even when the number of ISI taps scale as $K=2^{\kappa \log n}$ for any $\kappa \in [0,1)$, which we refer to as the ISI rate. This observation suggests that memory does not change the scale of the codebook derived for memoryless DTPC \cite{Salariseddigh22} and Gaussian channels \cite{Salariseddigh_ITW}.
    \item[$\blacklozenge$] \textbf{\textcolor{blau_2b}{Capacity Bounds}}: We derive DI capacity bounds for the DTPC with constant $K\geq 1$ and growing ISI $K = 2^{\kappa \log n}$, respectively. We show that for constant $K$, the proposed lower and upper bounds on $R$ are independent of $K$, whereas for growing ISI, they are functions of the ISI rate $\kappa$. Moreover, we show that optimizing $\kappa$ leads to an effective identification rate [bits/s] that  scales linearly with $n$, which is in contrast to the typical transmission rate [bits/s] that is independent of $n$.
    \item[$\blacklozenge$] \textbf{\textcolor{blau_2b}{Technical Novelty in The Capacity Proof}}: To obtain the proposed lower bound, the existence of an appropriate sphere packing within the input space, for which the distance between the centers of the spheres does not fall below a certain value, is guaranteed. This packing incorporates the effect of ISI as a function of $\kappa$. In particular, we consider the packing of hyper spheres inside a larger hyper cube, whose radius grows in both the codeword length $n$ and the ISI rate $\kappa$, i.e., $\sim n^{\frac{1+\kappa}{4}}$. For derivation of the upper bound, we assume that for given sequences of codes with vanishing error probabilities, a certain minimum distance between the codewords is asserted, where this distance depends on the ISI rate and decreases as $K$ grows.
\end{itemize}
\subsection{Organization}
The remainder of this paper is structured as follows. In Section~\ref{Sec.SysModel}, system model is explained and the required preliminaries regarding DI codes are established. Section~\ref{Sec.Res} provides the main contributions and results on the message identification capacity of the DTPC with ISI. Finally, Section~\ref{Sec.Summary} of the paper concludes with a summary and directions for future research.

\subsection{Notations}
We use the following notations throughout this paper: Calligraphic letters $\X,\Y,\Z,\ldots$ are used for finite sets. Lower case letters $x,y,z,\ldots$ stand for constants and values of random variables, and upper case letters $X,Y,Z,\ldots$ stand for random variables. Lower case bold symbol $\fx$ and $\fy$ stand for row vectors. Bold symbol $\boldsymbol{1}_{\bar{n}}$ indicates the all-one row vector of size $\bar{n}$. All logarithms and information quantities are for base $2$. The set of consecutive natural numbers from $1$ through $M$ is denoted by $[\![M]\!]$. The set of whole numbers is denoted by $\mathbb{N}_{0} \triangleq \{0,1,2,\ldots\}$. The set of non-negative real numbers is denoted by $\mathbb{R}_{+}$. The gamma function for non-positive integer $x$ is denoted by $\Gamma(x)$ and is defined as $\Gamma (x) = (x-1) !$, where $(x-1)! \triangleq (x-1) \times (x-2) \times \dots \times 1$. We use the small O notation, $f(n) = o(g(n))$, to indicate that $f(n)$ is dominated by $g(n)$ asymptotically, that is, $\lim_{n\to\infty} \frac{f(n)}{g(n)} = 0$. The big O notation, $f(n) = \mathcal{O}(g(n))$, is used to indicate that $|f(n)|$ is bounded above by $g(n)$ (up to constant factor) asymptotically, that is, $\limsup_{n\to\infty} \frac{|f(n)|}{g(n)} < \infty$. We use the big Omega notation, $f(n) = \Omega(g(n))$, to indicate that $f(n)$ is bounded below by $g(n)$ asymptotically, that is, $g(n) = \mathcal{O}(f(n))$. The $\ell_2$-norm and $\ell_{\infty}$-norm of vector $\fx$ are denoted by $\norm{\mathbf{x}}$ and $\norm{\mathbf{x}}_{\infty}$, respectively. Furthermore, we denote the $n$-dimensional hyper sphere of radius $r$ centered at $\fx_0$ with respect to the $\ell_2$-norm by $\S_{\fx_0}(n,r) = \big\{\fx\in\mathbb{R}_{+}^n : \norm{\fx-\fx_0} \leq r \big\}$. An $n$-dimensional cube with center $(\frac{A}{2},\ldots,\frac{A}{2})$ and a corner at the origin, i.e., $\mathbf{0} = (0,\ldots,0)$, whose edges have length $A$ is denoted by $\Q_{\f0}(n,A) = \big\{\fx \in \mathbb{R}_{+}^n : 0 \leq x_t \leq A, \forall \, t\in[\![n]\!] \big\}$. We denote the DTPC with $K$ ISI channel taps by $\P$.
\section{System Model and Preliminaries}
\label{Sec.SysModel}
In this section, we present the adopted system model and establish some preliminaries regarding DI coding.
\subsection{System Model}
We consider an identification-focused communication setup, where the decoder seeks to accomplish the following task: Determining whether or not an specific message was sent by the transmitter\footnote{\,\textcolor{mycolor5}{We assume that the transmitter does not know which message the decoder is interested in. This assumption is justified by the fact that otherwise, entire communication setting is specialized to transmission of only one indicator bit between Alice and Bob.}}; see Figure~\ref{Fig.E2E_Chain}. To attain this objective, a coded communication between the transmitter and the receiver over $n$ channel uses of an MC channel\footnote{\,\textcolor{mycolor5}{The proposed performance bounds works regardless of whether or not an specific code is used for communication, although proper codes may be required to approach such performance limits.}} is established. We consider the Poisson channel $\P$ which arises as a channel model in the context of MC for molecular counting receivers \cite{Gohari16}. Let $X\in\mathbb{R}_{\geq0}$ and $Y\in\mathbb{N}_0$ denote random variables (RVs) modeling the rate of molecule release by the transmitter and the number of molecules observed at the receiver, respectively. We consider a stochastic release model, where for the $t$-th channel use, the transmitter releases molecules with rate $x_t$ (molecules/second) over a time slot of $T_s$ seconds into the channel \cite{Gohari16}. These molecules propagate through the channel via diffusion and/or advection, and may even be degraded in the channel via enzymatic reactions \cite{Jamali19}. The receive is assumed to be equipped with a counting-type mechanism which is able to enumerate the number of received molecules observed in a determined volume.
\begin{figure}[t]
    \centering
    \vspace{.5in}
	\scalebox{1.08}{
\tikzstyle{0} = [draw, -latex']
\tikzstyle{Block_M} = [draw, top color = white, middle color = gray!20, rectangle, rounded corners, minimum height=2em, minimum width=1em]
\tikzstyle{Block1} = [draw, dashed, top color = white, middle color = orange_6b!30, rectangle, rounded corners, minimum height=3em, minimum width=7.6em]
\tikzstyle{Block2} = [draw, top color = white, middle color = cyan!30, rectangle, rounded corners, minimum height=2em, minimum width=5em]
\tikzstyle{Block3} = [draw, dashed, top color = white, middle color = orange_6b!30, rectangle, rounded corners, minimum height=3em, minimum width=8.4em]
\tikzstyle{Block1_1} = [draw, top color = white, middle color =  mycolor12!20, rectangle, rounded corners, minimum height=2em, minimum width=2em]
\tikzstyle{Block1_2} = [draw, top color = white, middle color = mycolor9!20, rectangle, rounded corners, minimum height=2em, minimum width=2em]

\tikzstyle{Block2_1} = [draw, top color = white, middle color = mycolor9!20, rectangle, rounded corners, minimum height=2em, minimum width=2em]
\tikzstyle{Block2_2} = [draw, top color = white, middle color = mycolor12!20, rectangle, rounded corners, minimum height=2em, minimum width=2em]
\begin{tikzpicture}[remember picture, overlay]
    \node[Block_M] (M) at (-6,0) {$\substack{\text{Messages}}$};
    \node[Block1, right=.64cm of M] (TX) {};
    \node[above = .1cm of TX] (TX_Text) {\text{\small TX}};
    \node[Block1_1, right=.85cm of M] (enc) {$\substack{\text{Enc}}$};
    \node[Block1_2, right=.3cm of enc] (particle_Gen) {$\substack{\text{Release}}$};
    \node[Block2, right=1cm of TX] (channel) {$\substack{\text{Diffus./Advec./Reac. Process.}}$};
    \node[above = .1cm of channel] (ch_Text) {\text{\small ISI Channel}};
    \node[Block3, right=.8cm of channel] (RX) {};
    \node[above = .1cm of RX] (RX_Text) {\text{\small RX}};
    \node[Block2_1, right= 1cm of channel] (rec) {$\substack{\text{Reception}}$};
    \node[Block2_2, right= .3cm of rec] (ver) {$\substack{\text{Dec}}$};
    \node[below=.45cm of ver] (Target) {$j$};
    \node[Block_M, right=.85cm of ver] (YN) {$\substack{\text{Yes\,/\,No}}$};
    \draw[-Triangle] (M) --node[above]{$i$} (enc);
    \draw[-latex] (enc) -- (particle_Gen);
    \draw[-Triangle] (particle_Gen) --node[above]{$\fu_i$} (channel);
    \draw[-Triangle] (channel) --node[above]{$\fY$} (rec);
    \draw[-latex] (rec) -- (ver);
    \draw[-Triangle] (ver) -- (YN);
    \draw[-latex] (Target) -- (ver);

\end{tikzpicture}
}
	\vspace{10mm}
	\caption{End-to-end transmission chain for DI communication in a generic MC system modelled as a DTPC. Relevant processes in the molecular channel include diffusion, advection, and chemical reactions. The transmitter maps message $i$ onto a codeword $\fc_i$. The receiver is provided with an arbitrary message $j$, and given the channel output vector $\fY$, it asks whether $j$ is identical to $i$ or not.}   
    \label{Fig.E2E_Chain}
\end{figure}

The channel memory is modelled by a length $K$ sequence of probability values, i.e., $\fp=[p_0,p_1,\ldots,p_{K-1}]$. The value $p_k$ in specifies the probability that a given molecule released by the transmitter at the beginning time slot $t$, is observed at the receiver during time slot $t+k$ and depends on the propagation environment (e.g., diffusion, advection, and reaction processes) and the reception mechanism (e.g., transparent, absorbing, or reactive receiver) as well as the distance between transmitter and receiver, see \cite[Sec.~III]{Jamali19} for the characterization of $\fp$ for various MC setups. Let $\rho_k \overset{\text{\scriptsize def}}{=} p_k T_s$ where the value $p_k \in (0,1]$ denotes the probability that a given molecule released by the transmitter at the beginning time slot $t$, is observed at the receiver during time slot $t+k$.

When the number of released molecules is large but only a small fraction of them arrives at the receiver, the relation of channel output $Y$ and input $X$ is characterized as follows \cite{Gohari16,Jamali19}:
\begin{align}
    \label{Eq.Poisson_Model}
    Y_t & = \text{Pois} \, \big( X_t^{\frho} + \lambda \big)
    \;,\,
\end{align}
where
\begin{align}
    X_t^{\frho} \overset{\text{\scriptsize def}}{=} \sum_{k=0}^{K-1} \rho_k X_{t-k} \,,\,    
\end{align}
is the mean number of observed molecules; see Figure~\ref{Fig.ISI}, due to the release of the transmitter and the constant $\lambda\in \mathbb{R}_{> 0}$ is the mean number of observed interfering molecules originating from external noise sources which employ the same type of molecule as the considered MC system. Let $\fx_t^* \overset{\text{\tiny def}}{=} (x_{t-K+1},\ldots,x_t)$ be the vector of the $K$ most recently released symbols. Then, the letter-wise transition probability law is given by
\begin{align}
    V(y_t|\fx_t^*) = \frac{e^{-\left( x_t^{\frho} + \lambda \right) } \left( x_t^{\frho} + \lambda \right)^{y_t}}{{y_t}!} \,.\,
\end{align}
We assume that different channel uses given any $K$ previous input symbols are statistically independent,  which is a valid assumption for, e.g., fully absorbing receivers \cite{Jamali19}. Therefore, for $n$ channel uses, the transition probability law is given by
\begin{align}
    \label{Eq.Poisson_Channel_Law}
    V^{\bar{n}}(\fy|\fx) = \prod_{t=1}^{\bar{n}} V(y_t|\fx_t^*) = \prod_{t=1}^{\bar{n}} \frac{e^{-\left( x_t^{\frho} + \lambda \right) } \left( x_t^{\frho} + \lambda \right)^{y_t}}{{y_t}!} \,,\,
\end{align}
where $\fx = (x_1,\dots,x_n)$ and $\fy = (y_1,\dots,y_{\bar{n}})$ denote the transmitted codeword and the received signal, respectively, with $\bar{n} = n+K-1$. We assume that $x_t = 0$ when $t > n$ or $t < 0$. The peak and average molecule release rate constraints on the codewords are
\begin{align}
    \label{Ineq.Const_X}
    0 \leq x_{t} \leq P_{\,\text{max}} \hspace{5mm} \text{and} \hspace{5mm} \frac{1}{n}\sum_{t=1}^{n} x_{t} \leq P_{\,\text{avg}} \,,\,
\end{align}
respectively, $\forall t\in[\![n]\!]$, where $P_{\,\text{max}} > 0$ and $P_{\,\text{avg}} > 0$ constrain the rate of molecule release per channel use and over the entire $n$ channel uses in each codeword, respectively.
\begin{remark}{(\textbf{Input Constraint Interpretation})}
    We note that while the average power constraint for the Gaussian channel is a non-linear (square) function of the symbols (signifying the signal energy), here for the DTPC, it is a linear function (signifying the number of released molecules)~\cite{Gohari16}.
\end{remark}
\begin{figure}[t]
    \centering
    \scalebox{1}{
\tikzstyle{Block_Start} = [draw, circle, radius = 1mm]
\tikzstyle{Block_CHL} = [draw, top color = white, middle color = cyan!30, rectangle, rounded corners, minimum height=2em, minimum width=5em]
\begin{tikzpicture}[scale=1]
    \begin{axis}[
    width=50mm,
    height=6cm,
    ymin=0,
    bar width=5mm,
    xtick={1,2,3,4,5},
    xticklabels={$u_{i,1}$,$u_{i,2}$,$u_{i,3}$,$u_{i,4}$,$u_{i,5}$},
    ytick style={draw=none},
    hide y axis,
    axis line style={draw=none},
    tick style={draw=none},
    every axis plot/.append style={
          ybar,
          bar shift=0pt,
          fill
        }
    ]
    \addplot [fill =gray!70, opacity = .9] coordinates {
        (1,1.8)
        };
    \addplot [fill =brown!90, opacity = .9] coordinates {
        (2,4)
        };
    \addplot [fill =mycolor7!90, opacity = .9] coordinates {
        (3,1.8)
        };
    \addplot [fill = orange_6b!90, opacity = .9] coordinates {
        (4,5.5)
        };
    \addplot [fill =mycolor13!90, opacity = .9] coordinates {
        (5,3)
        };
    \end{axis}
    \end{tikzpicture}

\hspace{14mm}


\begin{tikzpicture}
\begin{axis}[
    ytick style={draw=none},
    hide y axis,
    axis x line =bottom,
    axis y line=none,
    tick style={draw=none},
    width=21.5mm,
    height=3cm,
    ymin=0,
    bar width=4mm,
    xtick={1,2},
    xticklabels = {$p_0$,$p_1$},
    nodes near coords,
    every node near coord/.append style={font=\small},
    ybar,
    ]
\addplot [fill = blau_2b, opacity = .9 ] coordinates {
        (1,.5)
        (2,.5)
        };

\end{axis}
\end{tikzpicture}


\begin{tikzpicture}[remember picture, overlay]
\node at (-7,-3mm) (Start) {$\fu_i$};
\node at (+6.5,-3mm)  (End) {$\fY$};
\node[Block_CHL, xshift = -6mm, yshift= -3mm] (CHL) {$\text{DTPC with $2$-ISI}$};
\draw[-latex] (Start) -- (CHL);
\draw[-latex] (CHL) -- (End);
\end{tikzpicture}

\hspace{14mm}

\begin{tikzpicture}[scale=1]
    

    \begin{axis}[
    width=58mm,
    height=6cm,
    ybar stacked,
    ymin=0,  
    bar width=5mm,
    xtick={1,2,3,4,5,6},
    xticklabels = {$y_1$,$y_2$,$y_3$,$y_4$,$y_5$,$y_6$},
    ytick style={draw=none},
    hide y axis,
    axis line style={draw=none},
    tick style={draw=none},
    ]
    \addplot [fill=black] coordinates {
        (1,.12)
        (2,.12)
        (3,.12)
        (4,.12)
        (5,.12)
        (6,.12)
        };
        \addplot [fill = mycolor13!80, opacity = .9 ] coordinates {
        (1,0)
        (2,0)
        (3,0)
        (4,0)
        (5,1.2)
        (6,1.2)
        };
    \addplot [fill= orange_6b!90, opacity = .9] coordinates {
        (1,0)
        (2,0)
        (3,0)
        (4,2.2)
        (5,2.2)
        (6,0)
        };
    \addplot [fill = mycolor7!90, opacity = .9 ] coordinates {
        (1,0)
        (2,0)
        (3,.7)
        (4,.7)
        (5,0)
        (6,0)
        };
    \addplot [fill = brown!90, opacity = .9 ] coordinates {
        (1,0)
        (2,1.5)
        (3,1.5)
        (4,0)
        (5,0)
        (6,0)
        };
    \addplot [fill = gray!70, opacity = .9 ] coordinates {
        (1,.7)
        (2,.7)
        (3,0)
        (4,0)
        (5,0)
        (6,0)
        };
        \addplot [white, opacity = 0 ] coordinates {
        (1,2)
        (2,.5)
        (3,2)
        (4,0)
        (5,0)
        (6,0)
        };
    \end{axis}
    \end{tikzpicture}
}
    \vspace{7mm}
    \captionsetup{justification=justified}
    \caption{A DTPC with $2$-ISI channel with $\fp = (0.5,0.5)$. Channel takes an input sequence of non-negative real numbers and outputs a sequence with length $n + K - 1 = 5+2-1 = 6$ of integer numbers where each integer is a Poisson distributed random variable whose mean is sum of previous marked as accumulation of different colors. The constant interference $\lambda$ is depicted in black.}
    \vspace{-2mm}
    \label{Fig.ISI}
\end{figure}
\subsection{DI Coding for the DTPC}
The definition of a DI code for the DTPC $\P$ is given below.
\begin{definition}[ISI-Poisson DI Code]
\label{Def.ISI-Poisson-Code}
An $(n,\allowbreak M(n,R),\allowbreak K(n,\allowbreak \kappa), \allowbreak e_1, \allowbreak e_2)$ DI code for a DTPC $\P$ under average and peak molecule release rate constraints of $P_{\,\text{ave}}$ and $P_{\,\text{max}}$, respectively, and for integers $M(n,R)$ and $K(n,\kappa)$, respectively, where $n$ and $R$ are the codeword length and coding rate, respectively, is defined as a system $(\C,\mathscr{T})$, which consists of a codebook $\C=\big\{ \mathbf{c}_i \big\}_{i\in[\![M]\!]} \subset \mathbb{R}_{+}^n$, such that
\begin{align}
    0 \leq c_{i,t} \leq P_{\,\text{max}} \hspace{5mm} \text{and} \hspace{5mm} \frac{1}{n}\sum_{t=1}^{n} c_{i,t} \leq P_{\,\text{avg}} \,,\,
\end{align}
$\forall i\in[\![M]\!]$, $\forall t \in [\![n]\!]$, and a collection of decoding regions $\mathscr{T}=\{ \T_i \}_{i\in[\![M]\!]}$ with 
\begin{align}
    \bigcup_{i=1}^{M(n,R)}\T_i\subset\mathbb{N}_0^{\bar{n}} \,,\,
\end{align}
and $1 \leq K = K(n,R) < n$ being the number of ISI channel taps\footnote{\,\textcolor{mycolor5}{While in the definition of a DKI code, no specific restriction on the functional form of $K(n,R)$ is imposed, in our capacity results, it will be turned out that for at most a sub-linear form of $\sim n^{\kappa}$ for $0 \leq \kappa < 1$ being an arbitrary constant approaching one, non-trivial achievability results would be yielded.}}.
Given a message $i\in [\![M]\!]$, the encoder transmits $\mathbf{c}_i$, and the decoder's aim is to answer the following question: Was a desired message $j$ sent or not? There are two types of errors that may occur: Rejection of the true message (type I) or acceptance of a false message (type II). The corresponding error probabilities of the DI code $(\C,\mathscr{T})$ are given by
\begin{align}
    \label{Eq.ErrorConstraints}
    P_{e,1}(i) = 1 - \sum_{\fy \in \T_i} \hspace{-.7mm} V^{\bar{n}} ( \fy \, | \, \fc_i ) \hspace{5mm} \text{ and } \hspace{5mm} P_{e,2}(i,j) = \sum_{\fy \in \T_j} V^{\bar{n}} ( \fy \, | \, \fc_i ) \,,\,
\end{align}
and satisfy the following bounds $P_{e,1}(i) \leq e_1$ and $P_{e,2}(i,j) \allowbreak \leq e_2$, $\forall \allowbreak \, i,j \allowbreak \underset{i\neq j}{\in} [\![M]\!]$ and every $e_1, \allowbreak e_2>0$. A rate $R>0$ is called achievable if for every $e_1, \allowbreak e_2>0$ and sufficiently large $n$, there exists an $(n,\allowbreak M(n\allowbreak,R),\allowbreak K(n,\allowbreak \kappa), \allowbreak e_1, \allowbreak e_2)$ DI code. The DI capacity of the DTPC $\P$ is defined as the supremum of all achievable rates, and is denoted by $\mathbb{C}_{DI}(\P,M,K)$.
\end{definition}
\section{DI Capacity of the DTPC}
\label{Sec.Res}

In this section, we first present our main results, i.e., lower and upper bounds on the achievable identification rates for the DTPC with ISI. Subsequently, we provide the detailed proofs of these bounds.

\subsection{Main Results}
The DI capacity theorem for DTPC with ISI $\P$ is stated below.

\begin{theorem}
\label{Th.PDICapacity}
Consider the DTPC with ISI $\P$ and \emph{assume} that the number of ISI channel taps scales sub-linearly with codeword length $n$, i.e., $K(n,\kappa) = 2^{\kappa \log n}$, where $\kappa \in [0,1)$. Then the DI capacity of $\P$ subject to average and peak molecule release rate constraints of the form $n^{-1} \sum_{t=1}^n c_{i,t} \leq P_{\,\text{ave}}$ and $0 \leq c_{i,t} \leq P_{\,\text{max}}$, respectively, and a codebook of super-exponential scale, i.e., $M(n,R)=2^{(n\log n)R}$, is bounded by
\vspace{-1mm}
\begin{align}
    \label{Ineq.LU}
    \frac{1-\kappa}{4} \leq \mathbb{C}_{DI}(\P,M,K) \leq \frac{3}{2} + \kappa \,.\,
\end{align}
\end{theorem}
\begin{proof}
The proof of Theorem~\ref{Th.PDICapacity} consists of two parts, namely the achievability and the converse proofs, which are provided in Sections~\ref{Sec.Achiev} and \ref{App.Conv}, respectively.
\end{proof}

\begin{remark}
    The result in Theorem~\ref{Th.PDICapacity} comprises the following three special cases in terms of $K$:
    \begin{itemize}
        \item[$\blacklozenge$ \; \textbf{\textcolor{blau_2b}{Unit $K=1$}}:] This cases accounts for an ISI-free setup ($\kappa = 0$), which is valid when the symbol duration is large ($T_s\geq T_{\rm cir}$), and implies $K=1$ and $\kappa=0$. Thereby, $\bar{R}_{\rm eff}$ scales logarithmically with the codeword length $n$. This is in contrast to the transmission setting in which $\bar{R}_{\rm eff}$ is independent of $n$ (e.g., the well-known Shannon formula for the Gaussian channel). This result is known  in the identification literature \cite{Salariseddigh22,AD89}.
        \item[$\blacklozenge$ \; \textbf{\textcolor{blau_2b}{Constant $K>1$}}:] When $T_{s}$  is constant and $T_{s} < T_{\rm cir}$, we have constant $K>1$ which implies $\kappa\to0$ as $n\to\infty$. Surprisingly, our capacity result in Theorem~\ref{Th.PDICapacity} reveals that the bounds for the DTPC with memory are in fact identical to those for the memoryless DTPC given in \cite{Salariseddigh22}.
        \item[$\blacklozenge$ \; \textbf{\textcolor{blau_2b}{Growing $K$}}:] Our capacity results reveal that reliable identification is possible even when $K$ scales with the codeword length as $\sim 2^{\kappa \log n}$. Moreover, the impact of ISI rate $\kappa$ is reflected in the capacity lower and upper bounds in \eqref{Ineq.LU}, where the bounds respectively decrease and increase in $\kappa$. While the upper bound on $R_{\rm eff}$ increases in $\kappa$, too, the lower bound in \eqref{Ineq.Rbitps} suggests a trade-off in terms of $\kappa$, which is investigated in the Corollary~\ref{Corol.Optimum_Kappa}.
\end{itemize}    
\end{remark}

\begin{corollary}[Effective Identification Rate]
\label{Corol.R_eff}
Assuming $T_s = T_{\rm cir}/K = T_{\rm cir} 2^{-\kappa \log n}$, $\kappa \in [0,1)$, the effective identification rate, defined as 
\vspace{-3mm}
\begin{align}
    \label{Eq.R_eff_Def}
    \bar{R}_{\rm eff} \overset{\tiny \text{def}}{=} \frac{\log M(n,R)}{nT_s} \,
\end{align}
(in bits/s), under average and peak molecule release rate constraints is bounded by
\begin{align}
    \label{Ineq.Rbitps}
     \frac{\left(1-\kappa\right)n^{\kappa} \log n }{4T_{\rm cir}} \leq \bar{R}_{\rm eff} \leq \frac{\left(3+2\kappa\right)n^{\kappa} \log n}{2T_{\rm cir}} \,.\,
\end{align}
\end{corollary}
\begin{proof}
The proof follows directly by substituting the capacity results in Theorem~\ref{Th.PDICapacity} into the definition of the effective rate and making further mathematical simplifications.
\end{proof}

\begin{corollary}[Optimum ISI Rate]
    \label{Corol.Optimum_Kappa}
    The lower bound given in Corollary~\ref{Corol.R_eff} is maximized for the following ISI rate $\kappa_{\rm max}(n), n\in \mathbb{N}$, with
    \begin{align}
        \kappa_{\rm max}(n) = 1 - \frac{1}{\ln n} \,.\,
    \end{align}
    The above $\kappa_{\rm max}$ gives the following lower bound on $\bar{R}_{\rm eff}(n)$:
    \begin{align}
        \bar{R}_{\rm eff}(n) \geq \frac{\log e}{4eT_{\rm cir}} \cdot n \,.\,
    \end{align}
    Thereby,
    \begin{align}
        \liminf_{n \to \infty} \frac{\bar{R}_{\rm eff}(n) }{n} \geq \frac{\log e}{4eT_{\rm cir}}
        \,.\,
    \end{align}
\end{corollary}
\begin{proof}
The proof follows from differentiating the lower bound in Corollary~\ref{Corol.R_eff} with respect to $\kappa$ and equating it to zero.
\end{proof}
The effective identification rate $\bar{R}_{\rm eff}$ [bits/s] in \eqref{Eq.R_eff_Def} consists of two terms, namely the identification rate per symbol $\frac{\log M(n,R)}{n}$ [bits/symbol] (which decreases with $\kappa$ for the lower bound in \eqref{Ineq.LU}) and the symbol rate $\frac{1}{T_s}$ [symbol/s] (which increases with $\kappa$). The above corollary reveals that in order to maximize $\bar{R}_{\rm eff}$, it is optimal to set the trade-off for $\kappa$ such that the identification rate, i.e.,
\begin{align}
    \frac{\log M(n,R)}{n} & = \frac{(1 - \kappa_{\rm max})\log n}{4}
    \nonumber\\&
    = \frac{\log e}{4}
\end{align}
becomes independent of $n$ but the symbol rate, i.e.,
\begin{align}
    \frac{1}{T_s} = \frac{n}{eT_{\rm cir}}
\end{align}
linearly scales with $n$. As a result, in contrast to the typical transmission settings where the effective rate is independent of $n$, here, the effective identification rate $\bar{R}_{\rm eff}$ for the optimal $\kappa$ linearly grows in $n$.

\subsection{Achievability}
\label{Sec.Achiev}
The achievability proof consists of the following two main steps. \textbf{Step 1:} First, we propose a codebook construction and  derive an analytical lower bound on the corresponding codebook size using inequalities for sphere packing density. \textbf{Step~2:} Then, to prove that this codebook leads to an achievable rate, we propose a decoder and show that the corresponding type I and type II error rates vanished as $n \to \infty$.
\subsubsection*{Codebook construction}
Let
\begin{align}
    A = \min \left(P_{\,\text{ave}},P_{\,\text{max}} \right) \;.\,
\end{align}
In the following, we restrict ourselves to codewords that meet the condition $0 \leq x_t \leq A$, $\forall \, t \in [\![n]\!]$. We argue that this condition ensures both the average and the peak power constraints in \eqref{Ineq.Const_X}. In particular, when $P_{\,\text{ave}} \geq P_{\,\text{max}}$, then $A = P_{\,\text{max}}$ and the constraint $0 \leq x_t \leq A$ automatically implies that the constraint $\frac{1}{n} \sum x_t \leq P_{\,\text{ave}}$ is met, hence, in this case, the setup with average and peak power constraints simplifies to the case with only a peak power constraint. On the other hand, when $P_{\,\text{ave}} < P_{\,\text{max}}$, then $A = P_{\,\text{ave}}$ and by $0 \leq x_t \leq A$, $\forall \, t \in [\![n]\!]$, both power constraints are met, namely $\frac{1}{n} \sum x_t \leq P_{\,\text{ave}}$ and $0 \leq x_t \leq P_{\,\text{max}}$, $\forall \, t \in [\![n]\!]$. Hence, in the following, we restrict our considerations to a hyper cube with edge length $A$.

We use a packing arrangement of non-overlapping hyper spheres of radius $r_0 = \sqrt{n\theta_n}$ in a hyper cube with edge length $A$, where
\begin{align}
    \theta_n = \frac{a\sqrt{K}}{n^{\frac{1}{2}(1-b)}} = \frac{a}{n^{\frac{1}{2}(1-(b+\kappa))}} \;,\,
\end{align}
and $a>0$ is a non-vanishing fixed constant and $0 < b < 1$ is an arbitrarily small constant\footnote{\,\textcolor{mycolor5}{we recall that our achievability proof works for any $b\in(0,1)$; however, arbitrarily small values of $b$ are of interest since they result in the tightest lower bound.}}.

Let $\mathscr{S}$ denote a sphere packing, i.e., an arrangement of $M$ non-overlapping spheres $\S_{\fc_i}(n,r_0)$, $i\in [\![M]\!]$, that are packed inside the larger cube $\Q_{\f0}(n,A)$ with an edge length $A$. As opposed to standard sphere packing coding techniques \cite{CHSN13}, the spheres are not necessarily entirely contained within the cube. That is, we only require
\begin{wrapfigure}[17]{r}{0.4\textwidth}
\vspace{-\intextsep}
\vspace{-3mm}
  \begin{center}
    \scalebox{.78}{
\begin{tikzpicture}[scale=.6,rotate=45][thick]

\draw (-1.41,-1.41) circle (1cm);
\draw [fill=mycolor11, fill opacity=0.35] (-1.41,-1.41) circle (1cm);
\draw (0,0) circle (1cm);
\draw [fill=mycolor11, fill opacity=0.35] (0,0) circle (1cm);
\draw (1.41,1.41) circle (1cm);
\draw [fill=mycolor11, fill opacity=0.35] (1.41,1.41) circle (1cm);
\draw (+.52,-1.93) circle (1cm);
\draw [fill=mycolor11, fill opacity=0.35] (+.52,-1.93) circle (1cm);
\draw (1.93,-.52) circle (1cm);
\draw [fill=mycolor11, fill opacity=0.35] (1.93,-.52) circle (1cm);
\draw (-1.93,+.52) circle (1cm);
\draw [fill=mycolor11, fill opacity=0.35] (-1.93,.52) circle (1cm);
\draw (-.52,1.93) circle (1cm);
\draw [fill=mycolor11, fill opacity=0.35] (-.52,1.93) circle (1cm);
\node [fill=black, shape=circle, inner sep=.4pt] ($.$) at (2.45,-2.45) {};
\draw (2.45,-2.45) circle (1cm);
\draw [inner color=mycolor11,outer color=mycolor12, fill opacity=0.11] (2.45,-2.45) circle (1cm);
\node [fill=black, shape=circle, inner sep=.4pt] ($.$) at (2.81,2.81) {};
\draw (2.81,2.81) circle (1cm);
\draw [inner color=mycolor11,outer color=mycolor12, fill opacity=0.11] (2.81,2.81) circle (1cm);
\node [fill=black, shape=circle, inner sep=.4pt] ($.$) at (-2.45,2.45) {};
\draw (-2.45,2.45) circle (1cm);
\draw [inner color=mycolor11,outer color=mycolor12, fill opacity=0.11] (-2.45,2.45) circle (1cm);
\node [fill=black, shape=circle, inner sep=.4pt] ($.$) at (-2.81,-2.81) {};
\draw (-2.81,-2.81) circle (1cm);
\draw [inner color=mycolor11,outer color=mycolor12, fill opacity=0.11] (-2.81,-2.81) circle (1cm);

%

\foreach \s in {3}
{
\draw [thick] (-\s,-\s) -- (\s,-\s) -- (\s,\s) -- (-\s,\s) -- (-\s,-\s);
}

\node [fill=black, shape=circle, inner sep=.4pt] ($.$) at (0,0) {};
\draw [dashed] (0,0) -- (3,0) node [right,font=\large] {$A/2$};
\draw [dashed] (-2.81,-2.81) -- (-3.71,-2.83) node [left,font=\large] {$\sqrt{n\epsilon_n}$};
\draw [dashed] (3,-3) -- (-3,3) node [above left,font=\large] {$A\sqrt{n}$};
\end{tikzpicture}}
  \end{center}
  \captionsetup{justification=justified}
  \vspace{-4mm}
\caption{\footnotesize Illustration of a saturated sphere packing inside a cube, where small spheres of radius $r_0 = \sqrt{n\epsilon_n}$ cover a larger cube. Yellow colored spheres are not entirely contained within the larger cube, and yet they contribute to the packing arrangement. As we assign a codeword to each sphere center, the $1$-norm and arithmetic mean of a codeword are bounded by $A$ as required.}
\end{wrapfigure}
that the centers of the spheres are inside $\Q_{\f0}(n,A)$ and are disjoint from each other and have a non-empty intersection with $\Q_{\f0}(n,A)$. The packing density $\Updelta_n(\mathscr{S})$ is defined as the ratio of the saturated packing volume to the cube volume $\text{Vol}\left(\Q_{\f0}(n,A)\right)$, i.e.,
\begin{align}
    \Updelta_n(\mathscr{S}) \triangleq \frac{\text{Vol}\left(\bigcup_{i=1}^{M}\S_{\fc_i}(n,r_0)\right)}{\text{Vol}\left(\Q_{\f0}(n,A)\right)} \,.\,
    \label{Eq.DensitySphere}
\end{align}
Sphere packing $\mathscr{S}$ is called \emph{saturated} if no spheres can be added to the arrangement without overlap.

In particular, we use a packing argument that has a similar flavor as that observed in the Minkowski--Hlawka theorem for saturated packing \cite{CHSN13}.
Specifically, consider a saturated packing arrangement of 
\begin{align}
    \label{Eq.Union_Spheres}
    \bigcup_{i=1}^{M(n,R)} \S_{\fc_i}(n,\sqrt{n\theta_n})
\end{align}
spheres with radius $r_0=\sqrt{n\theta_n}$ embedded within cube $\Q_{\f0}(n,A)$. Then, for such an arrangement, we have the following lower \cite[Lem.~2.1]{C10} and upper bounds \cite[Eq.~45]{CHSN13} on the packing density 
\begin{align}
    \label{Ineq.Density}
    2^{-n} \leq \Updelta_n(\mathscr{S}) \leq 2^{-0.599n} \;.\,
\end{align}

In our subsequent analysis, we use the above lower bound which can be proved as follows: For the saturated packing arrangement given in \eqref{Eq.Union_Spheres}, there cannot be a point in the larger cube $\Q_{\f0}(n,A)$ with a distance of more than $2r_0$ from all sphere centers. Otherwise, a new sphere could be added which contradicts the assumption that the union of $M(n,R)$ spheres with radius $\sqrt{n\theta_n}$ is saturated. Now, if we double the radius of each sphere, the spheres with radius $2r_0$ cover thoroughly the entire volume of $\Q_{\f0}(n,A)$, that is, each point inside the hyper cube $\Q_{\f0}(n,A)$ belongs to at least one of the small spheres. In general, the volume of a hyper sphere of radius $r$ is given by \cite[Eq.~(16)]{CHSN13}
\begin{align}
    \text{Vol}\left(\S_{\fx}(n,r)\right) = \frac{\pi^{\frac{n}{2}}}{\Gamma(\frac{n}{2}+1)} \cdot r^{n} \,.\,
    \label{Eq.VolS}
\end{align}
Hence, if the radius of the small spheres is doubled, the volume of $\bigcup_{i=1}^{M(n,R)} \S_{\fc_i}(n,\sqrt{n\theta_n})$ is increased by $2^n$. Since the spheres with radius $2r_0$ cover $\Q_{\f0}(n,A)$, it follows that the original $r_0$-radius packing has a density of at least $2^{-n}$~\footnote{\,\textcolor{mycolor5}{We note that the proposed proof of the lower bound in \eqref{Ineq.Density} is non-constructive in the sense that, while the existence of the respective saturated packing is proved, no systematic construction method is provided.}}.
We assign a codeword to the center $\fc_i$ of each small sphere. The codewords satisfy the input constraint as $0 \leq c_{i,t} \leq A$, $\forall t \in [\![n]\!]$, $\forall i\in [\![M]\!]$, which is equivalent to
\begin{align}
    \label{Ineq.Norm_Infinity}
    \norm{\fc_i}_{\infty} \leq A \;.\,
\end{align}

Since the volume of each sphere is equal to $\text{Vol}(\S_{\fc_1}(n,r_0))$ and the centers of all spheres lie inside the cube, the total number of spheres is bounded from below by
\begin{align}
    \label{Eq.L_Achiev}
    M & = \frac{\text{Vol}\left(\bigcup_{i=1}^{M}\S_{\fc_i}(n,r_0)\right)}{\text{Vol}(\S_{\fc_1}(n,r_0))}
    \nonumber\\
    & = \frac{\Updelta_n(\mathscr{S}) \cdot
    \text{Vol}\left(\Q_{\f0}(n,A)\right)}{\text{Vol}(\S_{\fc_1}(n,r_0))}
    \nonumber\\&
     \geq 2^{-n} \cdot \frac{A^n}{\text{Vol}(\S_{\fc_1}(n,r_0))}
    \,,\,
\end{align}
where the first inequality holds by (\ref{Eq.DensitySphere}) and the second inequality holds by (\ref{Ineq.Density}).
The above bound can be further simplified as follows
\begin{align}
    \log M & \geq \log \left( \frac{A^n}{\text{Vol}\left(\S_{\fc_1}(n,r_0)\right)} \right) - n
    \nonumber\\
    & \stackrel{(a)}{=} n\log\left( \frac{A}{\sqrt{\pi} r_0 } \right)+\log \left(\frac{n}{2}! \right) - n
    \nonumber\\
    & \stackrel{(b)}{=} n \log A - n \log r_0 + \frac{1}{2} n \log n - n \log e + o(n) \;,\,
\end{align}
where $(a)$ exploits (\ref{Eq.VolS}) and $(b)$ holds by Stirling's approximation\footnote{\,\textcolor{mycolor5}{we recall that the packing of hyper spheres with with growing radius $\sim n^{\frac{1+2\kappa}{4}}$ in the codeword length inside a hyper cube with finite edge length is indeed counter-intuitive since the volume of such hyper sphere diverges to infinity as $n \to \infty$; see Appendix~\ref{App.SP_Diverging_Radius} for more details and explanations.}}. Now, for
\begin{align}
r_0 = \sqrt{n\theta_n}
= \sqrt{a} n^{\frac{1+b+\kappa}{4}} \,,\,
\end{align}
we obtain
\begin{align}
    \log M & \geq n \log \frac{A}{\sqrt{a}} - \frac{1}{4}\left(1+b+\kappa\right) \, n \log n + \frac{1}{2} n \log n - n \log e + o(n)
    \nonumber\\
    & = \left( \frac{1-(b+\kappa)}{4} \right) \, n \log n + n \Big( \log \frac{A}{e\sqrt{a}} \Big) + o(n)
    \;,\,
    \label{Eq.Log_L}
\end{align}
where the dominant term is of order $n \log n$. Hence, for obtaining a finite value for the lower bound of the rate, $R$, \eqref{Eq.Log_L} induces the scaling law of $M$ to be $2^{(n\log n)R}$. Therefore, we obtain
\begin{align}
    R & \geq \frac{1}{n\log n} \left[ \left( \frac{1-\left(b+\kappa\right)}{4} \right) n\log n + n \log \Big( \frac{A}{e\sqrt{a}} \Big) + o(n) \right] \;,\,
\end{align}
which tends to $\frac{1-\kappa}{4}$ when $n \to \infty$ and $b\rightarrow 0$.
\subsubsection*{Encoding}
Given message $i\in [\![M]\!]$, transmit $\fx=\fc_i$.
\subsubsection*{Decoding}
Let
\begin{align}
    \tau_n = c\rho_0^2\theta_n  = ac\rho_0^2n^{\frac{1}{2}((\kappa+b)-1)} \;,\,
    \label{Eq.Delta_n}
\end{align}
where $0 < b < 1$  is an arbitrarily small constant and $0<c<2$ is a constant. Before we proceed, for the sake of brevity of analysis, we introduce the following conventions:
\begin{itemize}
    \item Let $Y_t(i) \sim \text{Pois}( c_{i,t}^{\frho} + \lambda )$ denote the channel output at time $t$ given that $\fx=\fc_i$.
    \item Let $\fY(i)= (Y_1(i),\ldots, Y_{\bar{n}}(i))$
    \item Let $I_t^{\fx} \overset{\text{\tiny def}}{=} \lambda + \sum_{k=1}^{K-1} \rho_k x_{t-k}$.
    \item Let $\baron{y}_t(i) \overset{\text{\tiny def}}{=} y_t(i) - ( \rho_0 c_{i,t} + \lambda )$ where $y_t(i)$ is a realization of $Y_t(i)$.
\end{itemize}
\begin{remark}{(\textbf{Convoluted Symbol})}
    Observe that $c_{i,t}^{\frho} = \sum_{k=0}^{K-1} \rho_k c_{i,t-k}$ is only one symbol but is constructed from a linear combination of $K$ most recent symbols weighted by coefficients $\rho_k$.
\end{remark}
To identify whether message $j\in \M$ was sent, the decoder checks whether the channel output $\mathbf{y}$ belongs to the following decoding set:
\begin{align}
    \T_j = \left\{ \fy \in \Y^{\bar{n}} \;:\, \left| T(\fy;\fc_j) \right| \leq \tau_n \right\} \;,\,
\end{align}
where
\begin{align}
    T(\fy;\fc_j) = \frac{1}{\bar{n}} \sum_{t=1}^{\bar{n}} \left[ \left( y_t - ( \rho_0 c_{j,t} + \lambda ) \right)^2 - \left( \rho_0 c_{j,t} + I_t^{\fc_j} \right) - \left( I_t^{\fc_j} - \lambda
    \right)^2 \right]
\end{align}
is referred to as the decoding metric evaluated for observation vector $\fy$ and codeword $\fc_j$.

\subsubsection*{Error Analysis}
Fix $e_1,e_2 > 0$ and let $\zeta_0, \zeta_1 > 0$ be arbitrarily small constants.
Consider the type I errors, i.e., the transmitter sends $\fc_i$, yet $\fY\notin\T_i$. For every $i\in[\![M]\!]$, the type I error probability is given by
\begin{align}
    P_{e,1}(i) & = \Pr \left( \left| T(\fY(i);\fc_i) \right| > \tau_n \right) \,.\, 
    \label{Eq.TypeIError}
\end{align}
In order to bound $P_{e,1}(i)$, we apply Chebyshev's inequality, namely
\begin{align}
    \label{Eq.Chebyshev}
    \Pr\left(\left| T(\fY(i); \fc_i) - \mathbb{E} \left[ T(\fY(i); \fc_i) \right] \right| > \tau_n \right) \leq \frac{\text{Var} \left[ T(\fY(i);\fc_i) \right]}{\tau_n^2} \;.\,
\end{align}
Let us derive the expectation of the decoding metric as follows:
\begin{align}
    \mathbb{E} [ T(\fY(i); \fc_i) ] & \stackrel{(a)}{=} \frac{1}{\bar{n}} \sum_{t=1}^{\bar{n}} \mathbb{E}[\baron{Y}_t^2] - ( \rho_0 c_{i,t} + I_t^{\fc_i} ) - ( I_t^{\fc_i} - \lambda )^2
    \nonumber\\
    & \stackrel{(b)}{=} \frac{1}{\bar{n}} \sum_{t=1}^{\bar{n}} \text{Var}[\baron{Y}_t] + (\mathbb{E}[\baron{Y}_t])^2 - ( \rho_0 c_{i,t} + I_t^{\fc_i} ) - ( I_t^{\fc_i} - \lambda )^2
    \nonumber\\&
    \stackrel{(c)}{=} \frac{1}{\bar{n}} \sum_{t=1}^{\bar{n}} \text{Var} [ \baron{Y}_t(i) ]  + ( \mathbb{E} [ \baron{Y}_t(i) ] )^2 - ( \rho_0 c_{i,t} + I_t^{\fc_i} ) - ( I_t^{\fc_i} - \lambda )^2
    \nonumber\\
    & = 0 \,.\,
    \label{Eq.Expectation}
\end{align}
where $(a)$ follows from the linearity of expectation, $(b)$ holds since $\text{Var}[\baron{Y}_t] = \mathbb{E}[\baron{Y}_t^2] - (\mathbb{E}[\baron{Y}_t])^2$, and $(c)$ follows since $\text{Var} [ \baron{Y}_t(i) ] = \text{Var} [Y_t(i) ] = \rho_0 c_{i,t} + I_t^{\fc_i}$ and $\mathbb{E}[\baron{Y}_t] = (\rho_0 c_{i,t} + I_t^{\fc_i}) - (\rho_0 c_{i,t} + \lambda) = I_t^{\fc_i} - \lambda$.

Second, in order to compute the upper bound in \eqref{Eq.Chebyshev} we proceed to compute the variance of the decoding metric. Let us define
\begin{align}
   \label{Eq.PSI_Var}
   \psi_{\,\text{Var}} \overset{\text{\scriptsize def}}{=} \sum_{t=1}^{\bar{n}} \text{Var} [ ( \baron{Y}_t(i) )^2 ] \,.\,
\end{align}
we obtain
\begin{align}
    \text{Var}\left[ T(\fY(i);\fc_i)\right] = \frac{\psi_{\,\text{Var}}}{\bar{n}^2} \,,\,
\end{align}
since, conditioned on $\fc_i$, the channel outputs conditioned on the $K$ most recent input symbols are independent.
\begin{remark}{(\textbf{Output Correlation})}
     The pair ($\baron{Y}_t,\baron{Y}_{t'}$) is correlated as long as $\abs{t-t'} \leq K$; cf \eqref{Eq.Poisson_Model}. That is, output symbols with a distance of no more than $K$ are dependent on the same input symbols, and hence are correlated; cf \eqref{Eq.Poisson_Model}. However, in our error analysis, we always assume that the argument of the variance is conditioned on the given codeword $\fx=\fc_i$, hence, ($\baron{Y}_t(i),\baron{Y}_{t'}(i)$) are safely uncorrelated.
\end{remark}
Now, based on Appendix~\ref{App.Var_UB}, we provide an upper bound for the summand in \eqref{Eq.PSI_Var} as follows
\begin{align}
    \text{Var} \left[ ( \baron{Y}_t(i) )^2 \right] \leq 6(A + \lambda)^4 \left( 1 + e^{\frac{8}{\lambda}} \left( 1 + (A + \lambda) + (A + \lambda)^2 + (A + \lambda)^3 \right) \right)
    \,,\,
\end{align}
Thereby,
\begin{align}
    \psi_{\,\text{Var}} & \overset{\text{\scriptsize def}}{=} \sum_{t=1}^{\bar{n}} \text{Var} \left[ ( \baron{Y}_t(i) )^2 \right] 
    \nonumber\\&
    \leq 6\bar{n}(A + \lambda)^4 \left( 1 + e^{\frac{8}{\lambda}} \left( 1 + (A + \lambda) + (A + \lambda)^2 + (A + \lambda)^3 \right) \right)
    \nonumber\\&
    \overset{\text{\scriptsize def}}{=} \psi_{\,\text{Var}}^{\text{UB}}
    \,.\,
    \label{Ineq.Var}
\end{align}
Therefore, exploiting \eqref{Eq.Chebyshev}, \eqref{Eq.Expectation} and \eqref{Ineq.Var} we can bound the type I error probability in (\ref{Eq.TypeIError}) as follows
\begin{align}
    \label{Ineq.TypeI}
    P_{e,1}(i) & = \Pr \left( \left| T(\fY(i);\fc_i) \right| > \tau_n \right)
    \nonumber\\&
    \leq \frac{\psi_{\,\text{Var}}^{\text{UB}}}{\bar{n}^2\tau_n^2}
    \nonumber\\
    & = \frac{6(A + \lambda)^4 \left( 1 + e^{\frac{8}{\lambda}} \left( 1 + (A + \lambda) + (A + \lambda)^2 + (A + \lambda)^3 \right) \right)}{c^2\rho_0^4 a^2 \bar{n}n^{(\kappa+b)-1}}
    \nonumber\\
    & = \frac{6(A + \lambda)^4 \left( 1 + e^{\frac{8}{\lambda}} \left( 1 + (A + \lambda) + (A + \lambda)^2 + (A + \lambda)^3 \right) \right)}{c^2\rho_0^4 a^2 n^{\kappa+b}}
    \nonumber\\&
    \leq e_1 
    \,,\,
\end{align}
where the last equality follows from \eqref{Ineq.Var} and $n < \bar{n}$. Hence, $P_{e,1}(i) \leq e_1$ holds for sufficiently large $n$ and arbitrarily small $e_1 > 0$.

Next, we address type II errors, i.e., when $\fY\in\T_j$ while the transmitter sent $\fc_i$.
Then, for every $i,j\in[\![M]\!]$, where $i\neq j$, the type II error probability is given by
\begin{align}
    P_{e,2}(i,j) = \Pr \left( \left| T(\fY(i);\fc_j) \right| \leq \tau_n \right) \;.\,
\end{align}
where $T(\fY(i);\fc_j) = \beta - \alpha$ with
\begin{align}
    & \beta = \frac{1}{\bar{n}} \sum_{t=1}^{\bar{n}} \left( Y_t(i) - ( \rho_0 c_{i,t} + \lambda ) + \rho_0\left( c_{i,t} - c_{j,t} \right) \right)^2 \,,\,
    \\
    & \hspace{3cm} \alpha = \frac{1}{\bar{n}} \sum_{t=1}^{\bar{n}} \left[ \left( \rho_0 c_{j,t} + I_t^{\fc_j} \right) + \left( I_t^{\fc_j} - \lambda
    \right)^2 \right] \;.\,
\end{align}

Observe that term $\beta$ itself can be expressed by $\beta = \beta_1 + \beta_2$ where
\begin{align}
    & \beta_1 = \frac{1}{\bar{n}} \left[ \, \norm{ \fY(i) - \left( \rho_0 \fc_i + \lambda \boldsymbol{1}_{\bar{n}} \right)}^2 + \norm{\rho_0 \left( \fc_i - \fc_j \right)}^2\right] \,,\,
    \\
    & \hspace{3.5cm} \beta_2 = \frac{2\rho_0}{\bar{n}} \sum_{t=1}^{\bar{n}} \left( c_{i,t} - c_{j,t} \right)  \left( Y_t(i) - \left( \rho_0 c_{i,t} + \lambda \right) \right) \,.\,
\end{align}
Then, define the following events
\begin{align}
    & \mathcal{H}_i^j = \left\{ \left| \beta-\alpha  \right| \leq \tau_n  \right\}, 
    \\
    & \hspace{2cm} \E_0 = \left\{\left| \beta_2 \right| > \tau_n  \right\}, \quad
    \nonumber\\
    & \hspace{4cm} \E_1 = \left\{ \beta_1-\alpha \leq 2\tau_n \right\} \,.\,
\end{align}
Exploiting the reverse triangle inequality, i.e., $|\beta| - |\alpha| \leq |\beta - \alpha|$, 
we obtain the following upper bound on the type II error probability
\begin{align}
    P_{e,2}(i,j)
    & = \Pr\left( \mathcal{H}_i^j \right)
    \nonumber\\&
    = \Pr\left( |\beta - \alpha| \leq \tau_n \right)
    \nonumber\\
    & \leq \Pr\left( |\beta| - |\alpha| \leq \tau_n \right)
    \nonumber\\&
    \stackrel{(a)}{=} \Pr\left( \beta - \alpha \leq \tau_n \right),
\end{align}
where $(a)$ follows since $\alpha\geq0$ and $\beta\geq0$. 
Now, applying the law of total probability to event $\B = \big\{ \beta - \alpha \leq \tau_n \big\}$ over $\E_0$ and its complement $\E_0^c$, we obtain
\begin{align}
    P_{e,2}(i,j) & \leq \Pr \left( \B \cap \E_0 \right) + \Pr \left( \B \cap \E_0^c \right) 
    \nonumber\\& \overset{(a)}{\leq} \Pr\left(\E_0 \right) + \Pr \left( \B \cap \E_0^c \right) 
    \nonumber\\& \overset{(b)}{\leq} \Pr\left(\E_0 \right) + \Pr\left(\E_1 \right) \,,\,
\end{align}
where inequality $(a)$ follows from $\B \cap \E_0 \subset \E_0$ and inequality $(b)$ follows from $\Pr \left( \B \cap \E_0^c \right) \leq \Pr\left(\E_1 \right)$, which is proved in the following. Observe,
\begin{align}
   \Pr \left( \B \cap \E_0^c \right) 
   & = \Pr \left( \big\{ \beta - \alpha \leq \tau_n \big\} 
   \cap \left\{| \beta_2 | \leq \tau_n \, \right\} \right)
   \nonumber\\
   & = \Pr \left( \big\{ \beta_1 - \alpha \leq \tau_n-\beta_2 \big\} 
   \cap\big\{| \beta_2 | \leq \tau_n \, \big\} \right)
   \nonumber\\
   & \overset{(a)}{\leq} \Pr \left( \big\{ \beta_1 - \alpha \leq 2 \tau_n \big\} \right)
   \nonumber\\
   & = \Pr\left(\E_1 \right) \,,\,
\end{align}
where inequality $(a)$ holds since $\tau_n-\beta_2\leq 2\tau_n$ conditioned on $| \beta_2 | \leq \tau_n$.

We now proceed with bounding $\Pr\left(\E_0 \right)$. By Chebyshev's inequality, the probability of this event can be bounded as follows
\begin{align}
    \Pr(\E_0) & \leq \frac{\text{Var}\Big\{ \sum_{t=1}^{\bar{n}} \left( c_{i,t} - c_{j,t} \right) \big( Y_t - \left( \rho_0 c_{i,t} + \lambda \right) \big) \Big\}}{{\bar{n}^2\tau_n^2}/{(4\rho_0^2)}}
    \nonumber\\
    & = \frac{4 {\rho_0^2} \sum_{t=1}^{\bar{n}} (c_{i,t}-c_{j,t})^2\cdot\text{Var} [Y_t(i) ]}{\bar{n}^2\tau_n^2}
    \nonumber\\
    & = \frac{4{\rho_0^2} \sum_{t=1}^{\bar{n}}(c_{i,t}-c_{j,t})^2\cdot( \rho_0 c_{i,t}+I_t^{\fc_i})}{\bar{n}^2\tau_n^2}
    \nonumber\\
    & \leq \frac{4{\rho_0^2} (\rho_0 A+I_t^{\fc_i})\sum_{t=1}^{\bar{n}}(c_{i,t}-c_{j,t})^2}{\bar{n}^2\tau_n^2} 
    \nonumber\\
    & = \frac{4{\rho_0^2} (\rho_0 A+I_t^{\fc_i})\norm{\fc_i-\fc_j}^2}{\bar{n}^2\tau_n^2} \,.\,
\end{align}
Observe that
\begin{align}
    \norm{\fc_i - \fc_j}^2 \hspace{-.3mm} & \stackrel{(a)}{\leq} \hspace{-.5mm} \left(\norm{\fc_i} + \norm{\fc_j}\right)^2
    \nonumber\\
    & \hspace{-.3mm} \stackrel{(b)}{\leq} \hspace{-.5mm} \left(\sqrt{n} \norm{\fc_i}_{\infty} \hspace{-.3mm} + \hspace{-.3mm} \sqrt{n} \norm{\fc_j}_{\infty} \right)^2
    \nonumber\\&
    \hspace{-.4mm} \stackrel{(c)}{\leq} \hspace{-.5mm} \left(\sqrt{n} A \hspace{-.1mm} + \hspace{-.1mm} \sqrt{n} A \right)^2
    \nonumber\\&
     \hspace{-.3mm} = \hspace{-.3mm} 4nA^2 \hspace{-1.2mm} \;,\
\end{align}
where $(a)$ holds by the triangle inequality, $(b)$ follows since $\norm{\cdot} \leq \sqrt{n} \norm{\cdot}_{\infty}$, and $(c)$ is valid by (\ref{Ineq.Norm_Infinity}). Hence, we obtain
\begin{align}
    \Pr(\E_0) 
    & \leq \frac{16 n {\rho_0^2} (\rho_0  A+I_t^{\fc_i}) A^2}{\bar{n}^2\tau_n^2}
    \nonumber\\
    & \leq \frac{16 {\rho_0^2} (\rho_0  A+I_t^{\fc_i}) A^2}{n\tau_n^2}
    \nonumber\\ 
    & = \frac{16 (\rho_0  A+(K-1)A) A^2}{c^2\rho_0^2 a^2 n^{\kappa+b}}
    \nonumber\\ 
    & \leq \frac{16 (\rho_0  + n^{\kappa}) A^3}{c^2\rho_0^2 a^2 n^{\kappa+b}}
    \nonumber\\
    & \leq \frac{16 (\rho_0 + 1) A^3}{c^2\rho_0^2 a^2 n^{b}}
    \nonumber\\
    & \leq \zeta_0 \,,\,
\end{align}
for sufficiently large $n$, where $\zeta_0 > 0$ is an arbitrarily small constant.

We now proceed with bounding $\Pr\left(\E_1 \right)$ as follows. Based on the codebook construction, each codeword is surrounded by a sphere of radius $\sqrt{n\theta_n}$, that is 
 \begin{align}
    \label{Eq.distanceU}
     \norm{\fc_i - \fc_j}^2 \geq 4n\theta_n \;.\,
 \end{align}
Thus, we can establish the following upper bound for event $\E_1$:
\begin{align}
    \label{Ineq.E_1}
    \Pr(\E_1) & = \Pr \left( \frac{1}{\bar{n}} \left(\norm{\fY(i) \hspace{-.4mm} - \hspace{-.4mm} ({ \rho_0} \fc_i + \lambda\boldsymbol{1}_n)}^2 \hspace{-.5mm} + \hspace{-.5mm} \big\|\rho_0\left(\fc_i-\fc_j\right)\big\|^2 \hspace{-.4mm} - \hspace{-.4mm} \sum_{t=1}^{\bar{n}} \left( \rho_0 c_{j,t} + I_t^{\fc_j} \right) \hspace{-.4mm} + \hspace{-.4mm} \left( I_t^{\fc_j} - \lambda
    \right)^2 \right) \hspace{-.6mm} \leq 2\tau_n \right)
    \nonumber\\
    & \stackrel{(a)}{\leq} \Pr \left( \frac{1}{\bar{n}} \left( \norm{\fY(i) - \left( \rho_0 \fc_i + \lambda\boldsymbol{1}_n \right)}^2 - \sum_{t=1}^{\bar{n}} \left( \rho_0 c_{j,t} + I_t^{\fc_j} \right) + \left( I_t^{\fc_j} - \lambda
    \right)^2 \right) \leq 2(c-2)\rho_0^2\theta_n \right)
    \nonumber\\
    & = \Pr\left( \frac{1}{\bar{n}} \sum_{t=1}^{\bar{n}} \left( Y_t(i) - \left( \rho_0 c_{i,t} + \lambda \right) \right)^2 - \left( \rho_0 c_{j,t} + I_t^{\fc_j} \right) + \left( I_t^{\fc_j} - \lambda
    \right)^2 \leq 2(c-2)\rho_0^2\theta_n \right)
    \nonumber\\
    & \stackrel{(b)}{\leq} \frac{\text{Var}\left[ \frac{1}{\bar{n}} \sum_{t=1}^{\bar{n}} \left( Y_t(i) - \left( \rho_0 c_{i,t} + \lambda \right) \right)^2 \right]}{\left( 2(c-2)\rho_0^2\theta_n \right)^2}
    \nonumber\\
    & \stackrel{(c)}{\leq}
    \frac{\psi_{\,\text{Var}}^{\text{UB}}}{\left( 2\bar{n}(c-2)\rho_0^2\theta_n \right)^2}
    \nonumber\\
    & \leq \frac{6(A + \lambda)^4 \left( 1 + e^{\frac{8}{\lambda}} \left( 1 + (A + \lambda) + (A + \lambda)^2 + (A + \lambda)^3 \right) \right)}{4(c-2)^2 \rho_0^4 a^2 \bar{n} n^{(2\kappa+b)-1}}
    \nonumber\\
    & \leq \frac{6(A + \lambda)^4 \left( 1 + e^{\frac{8}{\lambda}} \left( 1 + (A + \lambda) + (A + \lambda)^2 + (A + \lambda)^3 \right) \right)}{4(c-2)^2 \rho_0^4 a^2 n^{2\kappa+b}}
    \nonumber\\
    & \leq \zeta_1 \;,\,
 \end{align}
for sufficiently large $n$, where $\zeta_1 > 0$ is an arbitrarily small constant. Here, $(a)$ follows from \eqref{Eq.distanceU} and (\ref{Eq.Delta_n}), $(b)$ holds by Chebyshev's inequality as given in (\ref{Eq.Chebyshev}), and $(c)$ follows by Appendix~\ref{App.Var_UB}. Therefore, 
$$P_{e,2}(i,j) \leq \Pr(\E_0) + \Pr(\E_1) \leq \zeta_0 + \zeta_1 \leq e_2 \,,\,$$
hence, $P_{e,2}(i,j) \leq e_2$ holds for sufficiently large $n$ and arbitrarily small $e_2 > 0$. We have thus shown that for every $e_1,e_2>0$ and sufficiently large $n$, there exists an $(n, M(n,R), K(n,\kappa), \allowbreak e_1, e_2)$ code.

\subsection{Converse Proof}
\label{App.Conv}
The proof of the converse is based on the following two steps. \textbf{Step~1:} First, we show in Lemma~\ref{Lem.DConv} that for any achievable rate (for which the type I and type II error rates vanish as $n\to\infty$), the distance between any selected entry of one codeword with any entry of another codeword should be at least larger than a threshold. \textbf{Step~2:} Then, using Lemma~\ref{Lem.DConv}, we derive an upper bound on the codebook size of achievable identification codes.

We start with the following lemma regarding the ratio of a function of the letters for every pair of codewords where such a function is defined as $d_{i,t} = \rho_0 c_{i,t} + I_t^{\fc_i} \;, \forall t \in [\![n]\!]$.
\begin{lemma}
\label{Lem.DConv}
Suppose that $R$ is an achievable rate for the DTPC $\P$. Consider a sequence of $(n, M(n,R), K(n,\kappa), \allowbreak e_1^{(n)}, \allowbreak e_2^{(n)})$ codes $(\C^{(n)},\T^{(n)})$ such that $e_1^{(n)}$ and $e_2^{(n)}$ tend to zero as $n\rightarrow\infty$. Then, given a sufficiently large $n$, the codebook 
$\C^{(n)}$ satisfies the following property.
For every pair of codewords, $\fc_{i_1}$ and $\fc_{i_2}$, there exists at least one letter $t \in [\![n]\!]$ such that 
\begin{align}
    \label{Eq.Converse_Lem}
    \left|1-\frac{\rho_0 c_{i_2,t} + I_t^{\fc_{i_2}}}{\rho_0 c_{i_1,t}+I_t^{\fc_{i_1}}}\right| > \theta'_n \,,\,
\end{align}
for all $i_1,i_2\in [\![M]\!]$, such that $i_1\neq i_2$,
with
\begin{align}
\label{Eq.epsilonn_p}
  \theta'_n = \frac{P_{\,\text{max}}}{Kn^{1+b}} = \frac{P_{\,\text{max}}}{n^{1+b+\kappa}} \,,\,
\end{align}
where $b>0$ is an arbitrarily small constant and $I_t^{\fc_{i_z}} \overset{\text{\tiny def}}{=} \lambda + \sum_{k=1}^{K-1} \rho_k c_{i_z,t-k} \;, z \in \{1,2\}$.
\end{lemma}
\begin{proof}
The proof is given in Appendix~\ref{App.Converse}.
\end{proof}
Next, we use Lemma~\ref{Lem.DConv} to prove the upper bound on the DI capacity. Observe that since
\begin{align}
    \label{Ineq.Shifted_Codeword}
    d_{i,t} & = \rho_0 c_{i,t} + I_t^{\fc_i} > \lambda \;,\,
\end{align}
Lemma~\ref{Lem.DConv} implies
\begin{align}
     \rho_0 \left| c_{i_1,t} - c_{i_2,t} \right| & = \left| d_{i_1,t} - d_{i_2,t} \right|
    \nonumber\\&
      \stackrel{(a)}{>} \theta'_n d_{i_1,t}
    \nonumber\\&
      \stackrel{(b)}{>} \lambda \theta'_n \;,\,
\end{align}
where $(a)$ follows by (\ref{Eq.Converse_Lem}) and $(b)$ holds by (\ref{Ineq.Shifted_Codeword}). Now, since $\norm{\fc_{i_1} - \fc_{i_2}} \geq \left| c_{i_1,t} - c_{i_2,t} \right|$, we deduce that the distance between every pair of codewords satisfies
\begin{align}
   \norm{\fc_{i_1} - \fc_{i_2}} > \frac{\lambda \theta'_n}{\rho_0} \;.\,
\end{align}
Thus, we can define an arrangement of non-overlapping spheres $\S_{\fc_i}(n,\frac{\lambda \theta'_n}{\rho_0})$, i.e., spheres of radius $\lambda \theta'_n$ that are centered at the codewords $\fc_i$. Since the codewords all belong to a hyper cube $\Q_{\f0}(n,P_{\,\text{max}})$ with edge length $P_{\,\text{max}}$, it follows that the number of packed small spheres, i.e., the number of codewords $M$, is bounded by
\begin{align}
    \label{Eq.L}
    M & = \frac{\text{Vol}\left(\bigcup_{i=1}^{M}\S_{\fc_i}(n,r_0)\right)}{\text{Vol}(\S_{\fc_1}(n,r_0))}
    \nonumber\\
    & = \frac{\Updelta_n(\mathscr{S}) \cdot
    \text{Vol}\left(\Q_{\f0}(n,P_{\,\text{max}})\right)}{\text{Vol}(\S_{\fc_1}(n,r_0))}
    \nonumber\\
    & \leq 2^{-0.599n} \cdot\frac{P_{\,\text{max}}^n}{\text{Vol}(\S_{\fc_1}(n,r_0))} \;,\,
\end{align}
where the last inequality follows from inequality (\ref{Ineq.Density}). Thereby,
\begin{align}
    \label{Eq.Converse_Log_L}
    \log M & \leq \log \left( \frac{P_{\,\text{max}}^n}{\text{Vol}\left(\S_{\fc_1}(n,r_0)\right)} \right) - 0.599n
    \nonumber\\
    & = n \log P_{\,\text{max}} - n \log r_0 - n \log \sqrt{\pi} + \frac{1}{2}n\log \frac{n}{2} - \frac{n}{2}\log e + o(n) - 0.599n \,,\;
\end{align}
where the dominant term is again of order $n \log n$. Hence, for obtaining a finite value for the upper bound of the rate, $R$, \eqref{Eq.Converse_Log_L} induces the scaling law of $M$ to be $2^{(n\log n)R}$. Hence, by setting
\begin{align}
    M(n,R) = 2^{(n\log n)R} \,,\,
\end{align}
and
\begin{align}
    r_0 = \frac{\lambda \theta'_n}{2\rho_0} = \frac{\lambda P_{\,\text{max}}}{2\rho_0 n^{1+b+\kappa}} \,,\,
\end{align}
we obtain
\begin{align}
    R & \leq \frac{1}{n\log n} \left[ n \log P_{\,\text{max}} - n \log r_0 - n \log \sqrt{\pi} + \frac{1}{2}n\log \frac{n}{2} - \frac{n}{2}\log e + o(n) - 0.599n \right]
    \nonumber\\
    & = \frac{1}{n\log n} \left[ \left( \frac{1}{2} + \left( 1 + b + \kappa \right) \right) \, n \log n - n \left( \log \frac{\lambda\sqrt{\pi e}}{2\rho_0} + 1.0599 \right)+ o(n) \right]
    \;,\,
\end{align}
which tends to $\frac{3}{2}+\kappa$ as $n \to \infty$ and $b \to 0$.  This completes the proof of Theorem~\ref{Th.PDICapacity}.
\section{Summary and Future Directions}
\label{Sec.Summary}
In this work, we studied the DI problem over the DTPC with $K$ number of ISI channel taps. We assumed that $K=K(n,\kappa)=2^{\kappa\log n} = n^{\kappa}$ where $\kappa \in [0,1)$ scales sub-linearly with the codeword length $n$. In practice, the DTPC exhibits memory \cite{Gohari16}, therefore, our results in this paper may serve as a model for event-triggered based tasks in the context of many practical MC applications. Especially, we obtained lower and upper bounds on the DI capacity of the DTPC with memory subject to average and peak power constraints with the codebook size of $M(n,R)=2^{(n\log n)R}=n^{nR}$. Our results for the DI capacity of the DTPC with memory revealed that the super-exponential scale of $n^{nR}$ is the appropriate scale for codebook size. This scale coincides the scale for codebook of memoryless DTPC and Gaussian channels \cite{Salariseddigh_IT,Salariseddigh_ITW} and stands considerably different from the traditional scales as in transmission and RI setups where corresponding codebooks size grows exponentially and double exponentially, respectively.

We show the achievability proof using a packing of hyper spheres and a distance decoder. In particular, we pack hyper spheres with radius $\sim n^{\frac{1+2\kappa}{4}}$ where $\kappa \coloneqq \log_n K \in [0,1)$ is the ISI rate, inside a larger hyper cube. While the radius of the spheres in a similar proof for Gaussian channels vanishes, as $n$ increases \cite{Salariseddigh_ITW}, the radius here similar to the case for memoryless DTPC \cite{Salariseddigh22} diverges to infinity. Yet, likewise as in \cite{Salariseddigh22} we can obtain a positive rate while packing a super-exponential number of spheres fulfilling the molecule release rates and error constraints. 

For the converse proof, we follow a similar approach as in our previous work for the memory-less DTPC \cite{Salariseddigh22}. In \cite{Salariseddigh22}, we established a minimum distance between each pair of shifted codewords when the amount of shift was the constant interference signal $\lambda>0$. Here, we let the value of shift vary according to the related codeword where it is lower bounded by $\lambda>0$. In general, the derivation here is more involved than the derivation in the Gaussian case \cite{Salariseddigh_ITW}. In our previous work on Gaussian channels with fading \cite{Salariseddigh_ITW}, the converse proof was based on establishing a minimum distance between each pair of codewords (with no shift). Here, on the other hand, we use the stricter requirement that the ratio of the letters of every two different shifted codewords is different from $1$ for at least one index.

The results presented in this paper can be extended in several directions, some of which are listed in the following as potential topics for future research works:
\begin{itemize}
    \item \textcolor{blau_2b}{\textbf{Continuous Alphabet Conjecture}}: Our observations for the codebook size of the memoryless DTPC and Gaussian channels \cite{Salariseddigh_ITW} lead us to conjecture that the codebook size for any continuous alphabet channel with/out memory is a super-exponential function, i.e., $2^{(n\log n)R}$. However, a formal proof of this conjecture remains unknown.
    \item \textcolor{blau_2b}{\textbf{Multi User}}: The extension of this study (point-to-point system) to multi-user scenarios (e.g., broadcast and multiple access channels) or multiple-input multiple-output channels may seems more relevant in complex MC nano-networks.
    \item \textcolor{blau_2b}{\textbf{Fekete's Lemma}}: Investigation of the behavior of the DI capacity in the sense of Fekete's Lemma \cite{Boche20}: To verify whether the pessimistic ($\underline{C} = \liminf_{n \to \infty} \allowbreak \frac{\log M(n,R)}{n \log n}$) and optimistic ($\overline{C} = \limsup_{n \to \infty} \frac{\log M(n,R)}{n \log n}$) capacities \cite{A06} coincide or not; see \cite{Boche20} for more details.
    \item \textcolor{blau_2b}{\textbf{Channel Reliability Function}}: A complete characterization of the asymptotic behavior of the decoding errors as a function of the codeword length for $0 < R < C$ requires knowledge of the corresponding channel reliability function (CRF) \cite{Boche21}. To the best of the authors’ knowledge, the CRF for DI has not been studied in the literature so far, neither for the Gaussian channel \cite{Salariseddigh_IT} nor the Poisson channel \cite{Salariseddigh_GC_IEEE,Salariseddigh_GC_arXiv,Salariseddigh22}.
    \item \textcolor{blau_2b}{\textbf{Explicit Code Construction}}: Explicit construction of DI codes with incorporating the ISI effect and the development of low-complexity encoding/decoding schemes for practical a designs where the associated efficiency of such codes can be evaluated with regard to to the our derived performance bounds in Section~\ref{Sec.Res}.
    \item \textcolor{blau_2b}{\textbf{ISI Gain}}: We have not exploited the ISI knowledge in the decoding procedure. We observed that for a DTPC with constant degree of ISI, capacity bounds coincide the bounds as of the memoryless DTPC. This observation suggest that testing a different decoding method which takes effect of ISI into account by conducting a symbol by symbol detection and exploits the previous $K$ input symbols might probably yields different and more accurate capacity bounds.
\end{itemize}
\appendix
\section{Upper bound For Variance}
\label{App.Var_UB}
Let $Y_t(i) \sim \text{Pois} \left( \rho_0 c_{i,t} + I_t^{\fc_i} \right)$ denote the channel output at time $t$ given that $\fx=\fc_i$. Recall that $\baron{y}_t(i) \overset{\text{\tiny def}}{=} y_t(i) - ( \rho_0 c_{i,t} + \lambda )$, then we have
\begin{align}
    \text{Var} \left[ ( \baron{Y}_t(i) )^2 \right] = & \text{Var} \left[ \left( Y_t - \left( \rho_0 c_{i,t} + \lambda \right) \right)^2 \right]
    \nonumber\\
    & \stackrel{(a)}{\leq} \mathbb{E} \left[ \left( Y_t(i) - \left( \rho_0 c_{i,t} + \lambda \right) \right)^4 \right] \nonumber\\&
    \stackrel{(b)}{=} \mathbb{E} \left[ Y_t^4(i) - 4Y_t^3(i) \left( \rho_0 c_{i,t} + \lambda \right) + 6Y_t^2(i) \left( \rho_0 c_{i,t} + \lambda \right)^2 - 4Y_t(i)\left( \rho_0 c_{i,t} + \lambda \right)^3 + \left( \rho_0 c_{i,t} + \lambda \right)^4 \right]
    \nonumber\\
    & \stackrel{(c)}{\leq} \mathbb{E} \left[ Y_t^4(i) \right] - 4\lambda \mathbb{E}\left[ Y_t^3(i) \right] + 6(\rho_0 c_{i,t} + \lambda)^2 \mathbb{E}\left[Y_t^2(i) \right] - 4(\rho_0 c_{i,t} + \lambda)^3 \mathbb{E}\left[ Y_t(i) \right] + (\rho_0 c_{i,t} + \lambda)^4
    \nonumber\\
    & \stackrel{(d)}{\leq} \mathbb{E} \left[ Y_t^4(i) \right] - 4\lambda \mathbb{E}\left[ Y_t^3(i) \right] + 6(A + \lambda)^2 \mathbb{E}\left[Y_t^4(i) \right] - 4(\rho_0 c_{i,t} + \lambda)^3 \mathbb{E}\left[ Y_t(i) \right] + (A + \lambda)^4
    \nonumber\\
    & \stackrel{(e)}{\leq} \mathbb{E} \left[ Y_t^4(i) \right] + 4\lambda \mathbb{E}\left[ Y_t^4(i) \right] + 6(A + \lambda)^2 \mathbb{E}\left[Y_t^4(i) \right] + 4(A + \lambda)^3 \mathbb{E}\left[ Y_t^4(i) \right] + (A + \lambda)^4
    \nonumber\\
    & \leq \mathbb{E} \left[ Y_t^4(i) \right] \left( 1 + 4\lambda + 6(A + \lambda)^2 + 4(A + \lambda)^3 \right) + (A + \lambda)^4
    \nonumber\\
    & \stackrel{(f)}{\leq} (A + \lambda)^4 e^{\frac{8}{\lambda}} \left( 1 + 4\lambda + 6(A + \lambda)^2 + 4(A + \lambda)^3 \right) + (A + \lambda)^4
    \nonumber\\
    & = (A + \lambda)^4 \left( 1 + e^{\frac{8}{\lambda}} \left( 1 + 4\lambda + 6(A + \lambda)^2 + 4(A + \lambda)^3 \right) \right)
    \nonumber\\
    & = 6(A + \lambda)^4 \left( 1 + e^{\frac{8}{\lambda}} \left( 1 + (A + \lambda) + (A + \lambda)^2 + (A + \lambda)^3 \right) \right)
    \,,\,
\end{align}
where $(a)$ follows from $\text{Var}\{Z\} \leq \mathbb{E} [Z_t^2]$ with $Z_t = \left( Y_t - \left( \rho_0 c_{i,t} + \lambda \right) \right)^2$, $(b)$ holds by the $8$-th order binomial expansion, $(c)$ follows by the linearity of the expectation operator, $(d)$ and $(e)$ follows from $c_{i,t} \leq A \,, \forall t \in [\![n]\!]$, and since expectation is an increasing function, that is, for integers $p,q$ we have
\begin{align}
    Y_t^p(i) < Y_t^q(i) \Rightarrow \mathbb{E} \left[ Y_t^p(i) \right] < \mathbb{E} \left[ Y_t^p(i) \right] \,,\,
\end{align}
$(f)$ holds by employing an upper bound on the non-central moment of a Poisson random variable with mean $\lambda_Z$ as follows (see \cite[Coroll.~1]{Ahle22})
\begin{align}
    \mathbb{E} \left[ Z^k \right] \leq \lambda_Z^k \, \exp\left\{\frac{k^2}{2\lambda_Z}\right\}
    \,.\,
\end{align}
\section{Volume of a Hyper Sphere With Growing Radius}
\label{App.SP_Diverging_Radius}
To solidify the idea of packing spheres within a hyper cube, we explain about the packing of hyper spheres with growing radius in the codeword length $n$. Despite the fact that radius of the hyper sphere's diverges to infinity as $n \to \infty$ as $\sim n^{\frac{1+\kappa}{4}}$, still the associated volume converges to zero super-exponentially inverse as of order $\sim n^{-\frac{(1+\kappa)n}{4}}$. This makes an accommodation of super-exponential number of such hyper spheres inside the hyper cube possible. The ratio of the spheres in our construction grows with $n$, as $\sim n^{\frac{1+\kappa}{4}}$. Volume of an $n$-dimensional \emph{unit}-hyper sphere, i.e., with a radius of $r_0=1$, tends to zero, as $n\to\infty$ \cite[Ch.~1, Eq.~(18)]{CHSN13}. 
Nonetheless, we observe that the volume still tends to zero for a radius of $r_0 = n^c$, where  $0 < c < \frac{1}{2}$. More precisely,
\begin{align}
    \lim_{n\to\infty} \text{Vol}\left(\S_{\fc_1}(n,r_0)\right) & = \lim_{n\to\infty}\frac{\pi^{\frac{n}{2}}}{\Gamma(\frac{n}{2}+1)}\cdot r_0^n
    \nonumber\\&
     = \lim_{n\to\infty} \frac{\pi^{\frac{n}{2}}}{\frac{n}{2}!}\cdot r_0^n
    \nonumber\\&
     = \lim_{n\to\infty} \left(\sqrt{\frac{2\pi}{n}}r_0\right)^n \,,\,
    \label{Eq.nSphere_Volume}
\end{align}
where the last equality follows by Stirling's approximation \cite[P.~52]{F66}, that is, $\log n! = n \log n - n \log e + o(n)$.
The last expression in (\ref{Eq.nSphere_Volume}) tends to zero for all $r_0 = n^c$ with $c \in (0,\frac{1}{2})$. Observe that when $n \to \infty$, the volume of a hyper cube $\Q_{\f0}(n,A)$ with edge length $A$ when $A<1$ tends to zeros, that is, $\lim_{n\to\infty} \text{Vol}\left(\Q_{\f0}(n,A)\right) = \lim_{n\to\infty} A^n = 0$.

Now, to count the number of spheres that can be packed inside the hyper cube $\Q_{\f0}(n,A)$, we derive the log-ratio of the volumes as follows
\begin{align}
    \label{Eq.Volume_Ratio}
    \log \left( \frac{\text{Vol}\left(\Q_{\f0}(n,A)\right)}{\text{Vol}\left(\S_{\fc_1}(n,r_0)\right)} \right)
    & = \log \left( \frac{A^n}{\pi^{\frac{n}{2}}{} r_0^n} \cdot \frac{n}{2}!\right)
    \nonumber\\
    & = n\log\left( \frac{A}{\sqrt{\pi} r_0 } \right)+\log \left(\frac{n}{2}! \right)
    \nonumber\\
    & = n \log A - n \log r_0  - n \log \sqrt{\pi} + \frac{1}{2}n\log \frac{n}{2} - \frac{n}{2}\log e + o(n)
    \nonumber\\
    & = \left(\frac{1}{2} - c \right) \, n \log n + n \left( \log \left(\frac{A}{\sqrt{\pi e}}\right)  - \frac{3}{2} \right) + o(n)
    \;,\,
\end{align}
where the last equality follows from $r_0 = n^c$. Now, since the dominant term in (\ref{Eq.Volume_Ratio}) involves $n\log n$, we deduce that codebook size should be $M(n,R) = 2^{(n\log n)R}$, thereby by (\ref{Eq.L_Achiev}) we obtain
\begin{align}
    \label{Ineq.Achiev_Rate_App}
    R & \geq \frac{1}{n\log n} \left[ \log \left( \frac{\text{Vol}\left(Q_{\f0}(n,A)\right)}{\text{Vol}\left(\S_{\fc_1}(n,r_0)\right)} \right) - n \right] 
    \nonumber\\
    & = \frac{1}{n\log n} \left[\left(\frac{1}{2} - c \right) \, n \log n + n \left( \log \left( \frac{A}{\sqrt{\pi e}} \right) - \frac{3}{2} \right) + o(n) \right] \,,\,
\end{align}
which tends to $\frac{1}{2} - c$ when $n \to \infty$. As a result, \eqref{Ineq.Achiev_Rate_App} induces that condition $c < \frac{1}{2}$ with $c$ not being arbitrary approaching $\frac{1}{2}$ to derive a meaningful (non-zero) lower bound. Since $c= \frac{1+\kappa}{4}$ we obtain
\begin{align}
    \frac{1+\kappa}{4} < \frac{1}{2} \Rightarrow \kappa < 1 \;.\,
\end{align}
\section{Proof of Lemma~\ref{Lem.DConv}}
\label{App.Converse}
In the following, we provide the proof of Lemma~\ref{Lem.DConv}. The method of proof is by contradiction, namely, we assume that the condition given in \eqref{Eq.Converse_Lem} is violated and then we show that this leads to a contradiction, namely, sum of the type I and type II error probabilities converge to one, i.e., $\lim_{n\to\infty} \left[ P_{e,1}(i_1) + P_{e,2}(i_1,i_2) \right] = 1$.

Recall that $Y_t(i) \sim \text{Pois}( c_{i,t}^{\frho} + \lambda )$ denote the channel output at time $t$ given that $\fx=\fc_i$. Fix $e_1,e_2 > 0$. Let $\eta, \delta > 0$ be arbitrarily small constants. Assume to the contrary that there exist two messages $i_1$ and $i_2$, where $i_1\neq i_2$, meeting the error constraints in \eqref{Eq.ErrorConstraints}, such that for all $t\in[\![n]\!]$, we have
\begin{align}
    \label{Ineq.Converse_Lem_Complement}
    \left|1-\frac{d_{i_2,t}}{d_{i_1,t}}\right| \leq \theta'_n \;,\,
\end{align}
where $d_{i_z,t}=\rho_0 c_{i_k,t} + I_t^{\fc_{i_z}},\,\,z=1,2$. In order to show contradiction, we will bound the sum of the two error probabilities, $P_{e,1}(i_1)+P_{e,2}(i_2,i_1)$, from below. To this end, define 
\begin{align}
    \label{Eq:Bi1}
       \R_{i_1} = \left\{\fy \in \T_{i_1} \,:\,
       \frac{1}{\bar{n}}\sum_{t=1}^{\bar{n}} y_t \leq  \rho_0 P_{\,\text{max}} + I_t^{\fc_{i_1}}  + \delta \right\} \;.\,
\end{align}
Then, observe that
    \begin{align}
    \label{Eq.Error_Sum_1}
    P_{e,1}(i_1)+P_{e,2}(i_2,i_1)
    & = 1- \sum_{\fy\in\T_{i_1}} V^{\bar{n}} \left( \fy \, \big| \, \fc_{i_1} \right) + \sum_{\fy\in\T_{i_1}} V^{\bar{n}} \left( \fy \, \big| \, \fc_{i_2} \right)
    \nonumber\\
    & \geq 1- \sum_{\fy\in\T_{i_1}} V^{\bar{n}} \left( \fy \, \big| \, \fc_{i_1} \right) + \sum_{\fy\in\T_{i_1} \cap \R_{i_1}} V^{\bar{n}} \left( \fy \, \big| \, \fc_{i_2} \right) \,.\,
    \end{align}
    
    Now, consider the sum over $\T_{i_1}$ in (\ref{Eq.Error_Sum_1}),
    \begin{align}
        \label{Ineq.Error_I_Complement}
        \sum_{\fy\in\T_{i_1}} V^{\bar{n}} \left( \fy \, \big| \, \fc_{i_1} \right) & = \sum_{\fy\in\T_{i_1}\cap\R_{i_1}} V^{\bar{n}} \left( \fy \, \big| \, \fc_{i_1} \right) + \sum_{\fy \in \T_{i_1}\cap\R_{i_1}^c} V^{\bar{n}} \left( \fy \, \big| \, \fc_{i_1} \right)
        \nonumber\\
        & \leq \sum_{\fy \in \T_{i_1}\cap\R_{i_1}} V^{\bar{n}} \left( \fy \, \big| \, \fc_{i_1} \right) + \Pr\left( \frac{1}{\bar{n}} \sum_{t=1}^{\bar{n}} Y_t(i_1) >  \rho_0 P_{\,\text{max}} + I_t^{\fc_{i_1}} + \delta \right) \,.\;
    \end{align}
    Next, we bound the probability on the right hand side of (\ref{Ineq.Error_I_Complement}) as follows
    \begin{align}
        & \Pr \left( \frac{1}{\bar{n}} \sum_{t=1}^{\bar{n}} Y_t(i_1) - \frac{1}{\bar{n}} \sum_{t=1}^{\bar{n}} \mathbb{E} [Y_t(i_1)] > \rho_0 P_{\,\text{max}} + \delta - \frac{1}{\bar{n}} \sum_{t=1}^{\bar{n}}  \mathbb{E} [Y_t(i_1)] \right)
        \nonumber\\
        & \overset{(a)}{\leq} \frac{\text{Var} \left[ \frac{1}{\bar{n}} \sum_{t=1}^{\bar{n}} Y_t(i_1) \right] }{\left( \rho_0 P_{\,\text{max}} + \delta - \frac{1}{\bar{n}} \sum_{t=1}^{\bar{n}} \mathbb{E} [Y_t(i_1)] \right)^2} \nonumber\\
        & \overset{(b)}{=} \frac{ \frac{1}{n^2} \sum_{t=1}^{\bar{n}} ( \rho_0 c_{i_1,t} + I_t^{\fc_{i_1}}) }{\left( \rho_0 P_{\,\text{max}} + \delta - \frac{1}{\bar{n}} \sum_{t=1}^{\bar{n}}  \rho_0 c_{i_1,t} + I_t^{\fc_{i_1}} \right)^2}
        \nonumber\\
        & \overset{(c)}{\leq} \frac{\rho_0 P_{\max} + I_t^{\fc_{i_1}}}{n\delta^2}
        \nonumber\\
        & \leq \frac{P_{\max} + \lambda + (K-1)A}{n\delta^2}
        \nonumber\\
        & \leq \frac{P_{\max} + \lambda + A}{n^{1-\kappa}\delta^2}
        \nonumber\\
        & \leq \eta
        \;,\,
        \label{Ineq.ErrorI_Complement}
   \end{align}
for sufficiently large $n$, where inequality $(a)$ follows from Chebyshev's inequality, for equality $(b)$, we exploited $\text{Var} [Y_t(i_1)] = \mathbb{E} [Y_t(i_1)] = \rho_0 c_{i_1,t} + I_t^{\fc_{i_1}}$, and for inequality $(c)$, we used the fact that $c_{i_1,t}\leq P_{\max} \,,\, \forall \, t\in[\![n]\!]$.

Returning to the sum of error probabilities in (\ref{Eq.Error_Sum_1}), exploiting the bound (\ref{Ineq.ErrorI_Complement}) leads to
\begin{align}
    \label{Eq.Error_Sum_2}
    P_{e,1}(i_1)+P_{e,2}(i_2,i_1)
     & \geq 
    1 - \sum_{\fy \in \T_{i_1} \cap \R_{i_1}} \left[ V^{\bar{n}} \left( \fy \, \big| \, \fc_{i_1} \right) - V^{\bar{n}} \left( \fy \, \big| \, \fc_{i_2} \right) \right] - \eta \,.\,
\end{align}
    Now, let us focus on the summand in the square brackets in (\ref{Eq.Error_Sum_2}). By (\ref{Eq.Poisson_Channel_Law}), we have
    \begin{align}
        \label{Ineq.Cond_Channel_Diff}
        V^{\bar{n}} \big( \fy \, \big| \, \fc_{i_1} \big) - V^{\bar{n}} \big( \fy \, \big| \, \fc_{i_2} \big) & = V^{\bar{n}} \big( \fy \, \big| \, \fc_{i_1} \big) \cdot \left[1 - {V^{\bar{n}} \left( \fy \, \big| \, \fc_{i_2} \right)}\,/\,{V^{\bar{n}} \left( \fy \, \big| \, \fc_{i_1} \right)}\right] \nonumber\\
        & = V^{\bar{n}} \big( \fy \, \big| \, \fc_{i_1} \big) \cdot \left[ 1 - \prod_{t=1}^{\bar{n}} e^{-(d_{i_2,t} - d_{i_1,t})} \left(\frac{d_{i_2,t}}{d_{i_1,t}}\right)^{y_t} \right]
        \nonumber\\
        & = V^{\bar{n}} \big( \fy \, \big| \, \fc_{i_1} \big) \cdot \left[ 1 - \prod_{t=1}^{\bar{n}} e^{-\theta'_n d_{i_1,t}} \left(1-\theta'_n\right)^{y_t} \right] \,,\,
    \end{align}
    where for the last inequality, we employed
    \begin{align}
        d_{i_2,t} - d_{i_1,t} \leq \left| d_{i_2,t} - d_{i_1,t} \right| \leq \theta'_n d_{i_1,t} \hspace{6mm}  \text{and} \hspace{6mm} 1-\frac{d_{i_2,t}}{d_{i_1,t}}\leq\left|1-\frac{d_{i_2,t}}{d_{i_1,t}}\right| \leq \theta'_n \,,\,
    \end{align}
    which follow from \eqref{Ineq.Converse_Lem_Complement}. 
    Now, we bound the product term inside the bracket as follows:
       \begin{align}
           \prod_{t=1}^{\bar{n}} e^{-\theta'_n d_{i_1,t}}\left(1-\theta'_n\right)^{y_t} & = e^{- \theta'_n\sum_{t=1}^{\bar{n}} d_{i_1,t}} \cdot \left( 1 - \theta'_n \right)^{\sum_{t=1}^{\bar{n}} y_t} \nonumber\\
           & \overset{(a)}{\geq} e^{-n\theta'_n \left( \rho_0 P_{\,\text{max}} + I_t^{\fc_{i_1}} \right)} \cdot \left( 1 - \theta'_n \right)^{n \left( \rho_0 P_{\,\text{max}} + I_t^{\fc_{i_1}} + \delta \right)}
            \nonumber\\
            & = e^{n\theta'_n\delta} \cdot e^{-n\theta'_n \left( \rho_0 P_{\,\text{max}} + I_t^{\fc_{i_1}} + \delta \right)} \cdot \left( 1 - \theta'_n \right)^{n \left( \rho_0 P_{\,\text{max}} + I_t^{\fc_{i_1}} + \delta \right)}
            \nonumber\\
            & \overset{(b)}{\geq} e^{n\theta'_n\delta} \cdot e^{-n\theta'_n \left( \rho_0 P_{\,\text{max}} + I_t^{\fc_{i_1}} + \delta \right)} \cdot \left( 1 - n\theta'_n \right)^{ \rho_0 P_{\,\text{max}} + I_t^{\fc_{i_1}}  + \delta }
            \nonumber\\
            & \geq e^{n\theta'_n \delta} \cdot f(n\theta'_n)
            \nonumber\\
            & \stackrel{(c)}{>} f(n\theta'_n) 
            \nonumber\\
            & \stackrel{(d)}{\geq} 1 - 3 \left( \rho_0 P_{\,\text{max}} + I_t^{\fc_{i_1}} + \delta \right) n\theta'_n
            \nonumber\\
            & = 1 - \frac{3 \left( \rho_0 P_{\,\text{max}} + \lambda + KA + \delta \right)P_{\,\max}}{n^{b+\kappa}}
            \nonumber\\
            & = 1 - \frac{3 \left( \rho_0 P_{\,\text{max}} + \lambda + An^{\kappa} + \delta \right)P_{\,\max}}{n^{b+\kappa}}
            \nonumber\\
            & = 1 - \frac{3 \left( \rho_0 P_{\,\text{max}} + \lambda + A+ \delta \right)P_{\,\max}}{n^{b}}
            \nonumber\\
            & \geq 1 - \eta \,.\,
    \end{align}
    for sufficiently large $n$ where $(a)$ follows since
    \begin{align}
        d_{i_1,t} \leq \rho_0 P_{\,\text{max}} + I_t^{\fc_{i_1}} \,,\, \forall \, t\in[\![n]\!] \,,\, \hspace{6mm} \text{and} \hspace{6mm}
        \sum_{t=1}^{\bar{n}} y_t \leq n \left( \rho_0 P_{\max} + I_t^{\fc_{i_1}} + \delta
        \right) \,,\,
    \end{align}
    where the latter inequality follows from $\fy\in\R_{i_1}$, cf. \eqref{Eq:Bi1}. For $(b)$, we used Bernoulli's inequality
    \begin{align}
        (1 - x)^r \geq 1 - rx \; \,,\, \; \forall \, x > -1 \; \,,\, \forall r>0 \,,\,
    \end{align}
    \cite[see Ch.~3]{Mitrinovic13}. For $(c)$, we exploited $e^{n\theta'_n \delta}> 1$ and the following definition: $f(x) = e^{-cx}(1-x)^c$
    with $c = I_t^{\fc_{i_1}} + \rho_0 P_{\,\text{max}}+\delta$. Finally, for $(d)$, we used the Taylor expansion $f(x) = 1-2cx + \mathcal{O}(x^2)$ to obtain the upper bound $f(x) \geq 1-3cx$ for sufficiently small values of $x$.


    Equation~(\ref{Ineq.Cond_Channel_Diff}) can then be written as follows
    \begin{align}
        \label{Ineq.Cond_Channel_Diff2}
        V^{\bar{n}} \big( \fy \, \big| \, \fc_{i_1} \big) - V^{\bar{n}} \big( \fy \, \big| \, \fc_{i_2} \big) & \leq
        V^{\bar{n}} \big( \fy \, \big| \, \fc_{i_1} \big) \cdot \left[ 1 - e^{\hspace{-.2mm} - \theta'_n \sum_{t=1}^{\bar{n}} d_{i_1,t}} \cdot \left( 1 - \theta'_n \right)^{\sum_{t=1}^{\bar{n}} y_t} \right]
        \nonumber\\&
        \leq \eta \cdot V^{\bar{n}} \big( \fy \, \big| \, \fc_{i_1} \big) \,.
    \end{align}
    Combining, (\ref{Eq.Error_Sum_2}), (\ref{Ineq.Cond_Channel_Diff}), and (\ref{Ineq.Cond_Channel_Diff2}) yields
    \begin{align}
        P_{e,1}(i_1) + P_{e,2}(i_2,i_1)
        & \overset{(a)}{\geq} 1 - \sum_{\fy\in\R_{i_1}} \left[ V^{\bar{n}} \left( \fy \, \big| \, \fc_{i_1} \right) - V^{\bar{n}} \left( \fy \, \big| \, \fc_{i_2} \right) \right] - \eta \nonumber \\ & = 1 - \sum_{\fy \in \R_{i_1}} \left[ \eta \cdot V^{\bar{n}} \left( \fy \, \big| \, \fc_{i_1} \right) \right] - \eta
        \nonumber\\&
         \overset{(b)}{\geq} 1 - 2\eta \;,\,
    \end{align}
    where for $(a)$, we replaced $\fy\in\R_{i_1}\cap\T_{i_1}$ by $\fy\in\R_{i_1}$ to enlarge the domain and for $(b)$, we used $\sum_{\fy \in \R_{i_1}} V^{\bar{n}} \left( \fy \, \big| \, \fc_{i_1}\right)\leq 1$. Clearly, this is a contradiction since the error probabilities tend to zero as $n\rightarrow\infty$. Thus, the assumption in (\ref{Ineq.Converse_Lem_Complement}) is false. This completes the proof of Lemma~\ref{Lem.DConv}.
\printendnotes
\bibliography{Lit}
\end{document}